\documentclass[a4paper,USenglish]{lipics-v2016}

\usepackage{amsmath,amsfonts,amssymb,latexsym,amsbsy,bm}
\usepackage{tikz}
\usepackage[ruled,vlined,commentsnumbered]{algorithm2e}
\usepackage{enumitem}
\usepackage{xspace} 
\usepackage{thm-restate}

\theoremstyle{plain}
\newtheorem{proposition}[theorem]{Proposition}

\newcommand{\N}{\mathbb{N}\xspace}
\newcommand{\bigO}{\mathcal{O}\xspace}

\newcommand{\ceil}[1]{\left \lceil #1 \right \rceil}
\newcommand{\alg}{\textsc{Alg}\xspace}
\newcommand{\opt}{\textsc{Opt}\xspace}
\newcommand{\algDivEx}{\textsc{Divide \& Explore}\xspace}
\newcommand{\agentset}{\mathcal{A}\xspace}
\newcommand{\agent}{A\xspace}

\SetEndCharOfAlgoLine{} 
\SetArgSty{textnormal} 

\title{Maximal Exploration of Trees with Energy-Constrained Agents
  \footnote{This work was partially supported by the ANR project ANCOR
    (\textsc{anr-14-ce36-0002-01}) and the DFG Priority Programme 1736
    ``Algorithms for Big Data''}}

\author[1]{Evangelos~Bampas}
\author[1]{J\'{e}r\'{e}mie~Chalopin}
\author[1]{Shantanu~Das}
\author[2]{Jan~Hackfeld}
\author[1]{Christina~Karousatou}

\affil[1]{LIS, Aix-Marseille Universit\'{e} \& CNRS, France,\\
  \texttt{\{evangelos.bampas,
    shantanu.das, jeremie.chalopin, christina.karousatou\}@lis-lab.fr}}
\affil[2]{School of Business and Economics, Humboldt University
  Berlin, Germany,\\ \texttt{jan.hackfeld@hu-berlin.de}}
\authorrunning{E.~Bampas, J.~Chalopin, S.~Das, J.~Hackfeld, and C.~Karousatou}

\Copyright{Evangelos~Bampas, J{\'e}r{\'e}mie~Chalopin, Shantanu~Das, Jan~Hackfeld, and Christina~Karousatou}

\begin{document}

\maketitle

\begin{abstract}
We consider the problem of exploring an unknown tree with a team of~$k$ initially colocated mobile agents. Each agent has limited energy and cannot, as a 
result, traverse more than~$B$ edges. The goal is to maximize the number of nodes collectively visited by all agents during the execution. Initially, 
the agents have no knowledge about the structure of the tree, but they gradually discover the topology as they traverse new edges. We assume that the 
agents can communicate with each other at arbitrary distances. Therefore the knowledge obtained by one agent after traversing an edge is 
instantaneously transmitted to the other agents. We propose an algorithm that
divides the tree into subtrees during the exploration process and
 makes a careful trade-off between breadth-first and depth-first exploration. We show that our algorithm is 3-competitive  compared to an optimal solution that we could obtain if we knew the map of the tree in advance. While it is easy to see that no algorithm can be better than 2-competitive, we give a non-trivial lower bound of $2.17$ on the competitive ratio of any online algorithm.
\keywords{graph exploration, mobile agents, online algorithm} 
\end{abstract}

\section{Introduction} \label{sec:intro}

The problem of exploration of an unknown graph by one or more agents  is a well known problem with many 
applications ranging from searching the internet to physical exploration of unknown terrains using mobile sensor robots. 
Most results on exploration 
algorithms focus on minimizing the exploration time or memory requirements for the agents. For a brief survey of such results see~\cite{Das2013}.

We study the exploration problem under the very natural constraint that
agents have limited energy resources and movement consumes energy.
We model this constraint by bounding the number of edges that an agent 
can traverse by an integer~$B$ (henceforth called the \emph{energy budget} of the agent).  
 A similar restriction was considered in the \emph{piecemeal 
exploration} problem~\cite{DBLP:journals/iandc/AwerbuchBRS99}, where the agent could refuel by going back to its starting location. Thus, the 
exploration could be performed by 
a 
single agent using a sequence of tours starting and ending at the root vertex. On the other hand, \cite{DBLP:conf/arcs/DyniaKS06} 
and~\cite{DBLP:conf/sirocco/0001DK15} studied 
exploration without refueling, using multiple agents with the objective of minimizing the energy budget per agent, or the number of agents needed 
for a 
fixed budget. 

In this paper, we drop the requirement that the graph needs to be completely explored by the agents and instead focus on exploring the maximum number of nodes with a fixed given number $k$ of initially colocated agents with fixed energy budgets~$B$.
The tree is initially unknown to the agents and 
its topology is gradually discovered as the agents visit new nodes. We assume that the agents can communicate globally, i.e., an updated map of the tree is transmitted instantaneously to all other agents. We measure the performance of an algorithm for this problem by the standard tool of competitive analysis, i.e., we compare
a given online algorithm to an optimal offline algorithm which has a complete map of the tree in advance.

\subparagraph{Our result.}
 
 The challenge in designing a good exploration algorithm for our problem is to balance between 
sending agents in a depth-first manner to 
avoid visiting the same set of vertices too often and 
exploring the tree in a breadth-first manner to make sure that there is no large set 
of vertices close to the root that was missed by the online algorithm. Our algorithm achieves this by
maintaining a set of edge-disjoint subtrees of the part of the tree that is already explored and by iteratively sending an agent from the root to the 
subtree with the highest root. We prove that the algorithm is 3-competitive, i.e., an optimal offline algorithm which knows the tree in advance can 
explore at most three times as many vertices as our algorithm. 
We also show that the analysis is tight by giving a sequence of instances showing that the algorithm is not better than 3-competitive. 

We complement this positive result by showing that
no online algorithm can be better than $2.17$-competitive. 
The proof of this general lower bound is based on an adaptive adversary that
forces the online algorithm to spend a lot of energy if it completely wants to explore certain subtrees while preventing it from discovering some vertices close to the root.

\subparagraph{Further related work.}

Graph exploration has received much attention in the literature as it is a basic subtask for solving other more complex problems. For the problem of graph exploration by a single agent a very common optimization objective is to minimize the exploration time or equivalently the number of edges traversed. Depth First Search is a very simple exploration algorithm that requires at most $2m$ steps for exploring a graph with $m$ edges. 
In \cite{PanPe99}, Panaite and Pelc improved this result by presenting an algorithm that requires $m +3n$ steps for exploring a graph of $n$ nodes and $m$ edges. Graph exploration has been studied for the case of weighted graphs \cite{MEGOW201262}, as well as for the much harder case of directed graphs \cite{AlbH00,Deng,Fleischer2005,Wattenhofer}. The objective of minimizing the memory has been investigated in \cite{Lgmem,Disser:2016,FraIPPP05}.

For the case of multiple agents Fraigniaud et al. in \cite{FraGKP06} showed that the problem of minimizing the exploration time of $k$ collaborating agents is NP-hard even for tree topologies. The authors also proposed an  $\bigO(k/ \log k)$-competitive algorithm for collaborative graph exploration and gave a lower bound on the competitive ratio of $\Omega(2 - 1/k)$. This lower bound was later improved to $\Omega(\log k / \log \log k)$ in~\cite{Dynia2007}.
The lower and upper bounds for the problem of collaborative exploration have been further investigated in \cite{Disser15,Disser17,Ortolf2014}. 

Energy aware graph exploration was first considered by Betke et al. in \cite{Betke95}. The authors studied exploration of grid graphs by an agent who can return to its starting node $s$ for refueling (piecemeal exploration). The agent is given an upper bound $(2 +\alpha)r$ on the number of edge traversals it can make before returning to $s$, where $\alpha$ is some positive constant and $r$ is the distance to the furthest node from $s$. They presented an  $\bigO(m)$ algorithm for exploration of grid graphs with rectangular obstacles. In \cite{DBLP:journals/iandc/AwerbuchBRS99} an algorithm  for piecemeal exploration of general graphs was proposed requiring a nearly linear number of edge traversals. Finally, an optimal algorithm for piecemeal exploration of weighted graphs requiring only $\Theta(m)$ edge traversals was presented in \cite{Duncan01}. In \cite{arxiv17} the authors studied the exploration of weighted trees and showed that in this case, the decomposition of any DFS traversal of a weighted tree into a sequence of closed routes of length at most $B$ provides a constant-competitive solution with respect to the number of routes as well as the total energy consumption. 

When refueling is not allowed, multiple agents may be needed to explore even graphs of restricted diameter. Dynia et al. \cite{DBLP:conf/arcs/DyniaKS06} studied collaborative exploration with return for the case where the number of agents is fixed and the goal is to minimize the amount of energy $B$ required by each agent. They presented an $8$-competitive algorithm for trees and showed a lower bound of $1.5$ on the competitive ratio for any deterministic algorithm. The upper bound was later improved to $4 -2 /k$ in \cite{Dynia2007}. The authors in \cite{DBLP:conf/sirocco/0001DK15} considered tree exploration with no return for the case where the amount of available energy to the agents is fixed and the goal is to minimize the number of agents used. They presented an algorithm with a competitive ratio of $\bigO(\log B)$ for the case that the agents need to meet in order to communicate and showed that this is best possible.
 
\section{Terminology and Model}

We consider a set $\agentset$ of $k$ distinct agents initially located at the root~$r$ of an undirected, initially unknown tree~$T$. The edges at every 
vertex~$v$ in~$T$ have locally distinct edge labels $0,\ldots, \delta_v-1$, where $\delta_v$ is the degree of $v$. These edge labels are referred to as the \emph{local port numbers} at~$v$. We assume, without loss of generality, that the local port number of the edge leading back to the root~$r$ is~$0$ for any 
vertex~$v\neq r$ in~$T$. 
Otherwise, every agent internally swaps the labels of the edge leading back to the root and the label $0$ for every vertex $v\neq r$.

For any vertex~$v$ in~$T$, we let $d(v)$ be the depth of $v$ in $T$. The induced subtree with root~$v$ containing $v$ and all vertices below $v$ in $T$ is further denoted by $T(v)$. For a subtree $S$ of $T$, we write $r_S$ to denote the root of $S$, i.e., the unique vertex contained in $S$ having the smallest depth in $T$. Moreover, $|S|$ denotes the number of vertices in~$S$.

The tree is initially unknown to the agents, but they learn the map of the tree as they traverse new edges. Each time an agent arrives at a new vertex, 
it learns the local port number of the edge through which it arrived, as well as the degree of the vertex. We assume that agents can communicate at 
arbitrary distances, so the updated map of the tree, including all agent 
positions, is instantaneously available to all agents (global communication). Each agent has limited energy~$B$ and it consumes one unit of 
energy for every edge that it traverses.

The goal is to design an algorithm~\alg that maximizes the total number of distinct vertices visited by the agents. For a given instance $I=\left\langle T,r,k,B\right\rangle$, where $T$ is a tree, $r$ is the starting vertex of the agents, $k$ is the number of agents, and $B$ is the energy budget of each agent, 
let $|\alg(I)|$ denote the total number of distinct vertices visited by the agents using algorithm~\alg on the instance~$I$. We measure the performance of an algorithm~$\alg$ by the 
competitive ratio $\rho_{\alg} = \sup_I \frac{|\opt(I)|}{|\alg(I)|}$, where $|\opt(I)|$ is the maximum number of distinct vertices of $T$ that can be explored by the agents using an optimal offline algorithm \opt, i.e., an algorithm with full initial knowledge of the instance~$I$. When the considered instance~$I$ is clear from the context, we simply write $|\alg|$ instead of $|\alg(I)|$ and $|\opt|$ instead of~$|\opt(I)|$. 

\section{An Algorithm for Maximal Tree Exploration}

Let us assume that we do a depth-first search of the whole tree $T$ and always choose the smallest label $l>0$ to
an unexplored vertex. We call this algorithm L-DFS. We further denote the sequence $(r,v_1), (v_1,v_2) \ldots,(v_m,r)$ of directed edges obtained by directing every undirected edge of $T$ that the agent traversed in the direction in which the agent traversed the edge in the L-DFS traversal 
 the \emph{L-DFS sequence} of $T$. Note that every undirected edge $\{v,w\}$ of the tree $T$ appears as $(v,w)$ and $(w,v)$ in this sequence. 
 Similarly, we call a depth-first search of $T$ that always chooses the largest label $l>0$ to an unexplored vertex an R-DFS and the corresponding sequence of directed edges an \emph{R-DFS sequence}. 
 Note that the R-DFS sequence of the edges in $T$ is obtained by reversing the order of edges of the L-DFS sequence and changing every edge $(v,w)$ to $(w,v)$. A concrete implementation of both the algorithm L-DFS and the algorithm R-DFS is given in Appendix~\ref{app:LR-DFS}. 
 
 We call a consecutive subsequence of an L-DFS or R-DFS sequence a \emph{substring}.
For an induced subtree $T(v)$ of $T$, the L-DFS sequence of $T(v)$ is simply a substring of the L-DFS sequence of $T$. For a subtree $S$ we define the \emph{leftmost} unexplored vertex as the unexplored vertex in $S$ which is incident to the first edge in the L-DFS sequence of $S$ leading to an unexplored vertex and the \emph{rightmost} unexplored vertex as the unexplored vertex in $S$ which is incident to the first edge in the R-DFS sequence of $S$ leading to an unexplored vertex.

We further say that an agent~$\agent$ performing an L-DFS \emph{covers} at least $s$~edges $(v_1,v_2), \ldots, (v_s,v_{s+1})$ of the L-DFS sequence of~$T$, if $\agent$ consecutively visits $v_1, v_2,\ldots, v_s,v_{s+1}$ in this order  and the sequence $(v_1,v_2),\ldots, (v_s,v_{s+1})$ is a substring of the L-DFS sequence of $T$. 
Similarly, we say that an agent~$\agent$ performing an R-DFS \emph{covers} at least $s$ edges $(v_1,v_2),\ldots, (v_s,v_{s+1})$ of the L-DFS sequence of $T$, if $\agent$ consecutively visits $v_{s+1}, v_s,\ldots, v_2,v_1$ in this order and the sequence $(v_1,v_2),\ldots, (v_s,v_{s+1})$ is a substring of the L-DFS sequence of $T$. Note that two agents $A_1$ and $A_2$ may traverse the same edge in the same direction, but still cover two distinct sets of directed edges of the L-DFS sequence, if one agent performs an L-DFS and the other agent an R-DFS.

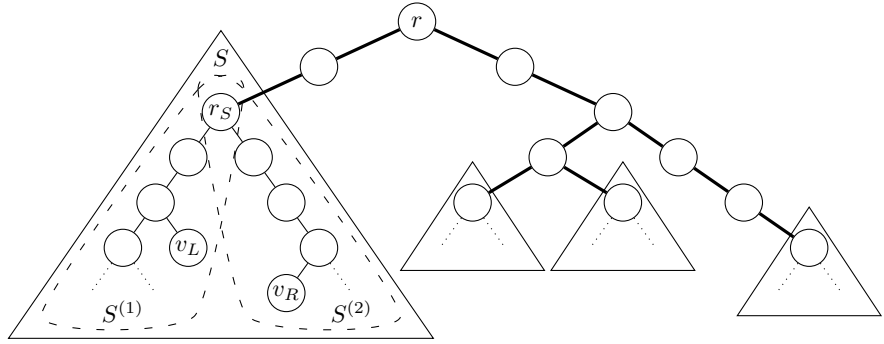
\begin{figure}
\begin{center}
\usetikzlibrary{arrows,intersections}
\tikzstyle{gnode}=[circle,draw,minimum size=2em,scale=0.7]
\tikzstyle{lnode}=[scale=0.9]
\tikzstyle{bedge}=[very thick]
\newcommand*{\xf}{0.43}
\newcommand*{\yf}{0.6}
\begin{tikzpicture}[scale=1]
\node[gnode] (r) at (0,0) {};
\node[lnode] (l0) at (0,0)  {$r$};	

\node[gnode] (v1) at (-3*\xf,-\yf) {};
\node[gnode] (v2) at (-6*\xf,-2*\yf) {};
\node[lnode] (l0) at (-6*\xf,-2*\yf)  {$r_S$};	
\node[gnode] (v3) at (-7*\xf,-3*\yf) {};
\node[gnode] (v4) at (-5*\xf,-3*\yf) {};
\node[gnode] (v5) at (-8*\xf,-4*\yf) {};
\node[gnode] (v6) at (-9*\xf,-5*\yf) {};
\node[gnode] (v7) at (-7*\xf,-5*\yf) {};
\node[lnode] (l7) at (-7*\xf,-5*\yf) {$v_L$};
\node[gnode] (v8) at (-4*\xf,-4*\yf) {};
\node[gnode] (v9) at (-3*\xf,-5*\yf) {};
\node[gnode] (v10) at (-4*\xf,-6*\yf) {};
\node[lnode] (l10) at (-4*\xf,-6*\yf) {$v_R$};

\draw[bedge](r)--(v1);
\draw[bedge](v1)--(v2);
\draw(v2)--(v3);
\draw(v3)--(v5);
\draw(v5)--(v6);
\draw(v5)--(v7);
\draw[dotted](v6)-- ++ (-\xf,-\yf);
\draw[dotted](v6)--++ (\xf,-\yf);

\draw(v2)--(v4);
\draw(v4)--(v8);
\draw(v8)--(v9);
\draw(v9)--(v10);
\draw[dotted](v9)--++ (\xf,-\yf);

\node[gnode] (w1) at (3*\xf,-\yf) {};
\node[gnode] (w2) at (6*\xf,-2*\yf) {};
\node[gnode] (w3) at (4*\xf,-3*\yf) {};
\node[gnode] (w4) at (8*\xf,-3*\yf) {};
\node[gnode] (w5) at (1.7*\xf,-4*\yf) {};
\node[gnode] (w6) at (6.3*\xf,-4*\yf) {};
\node[gnode] (w8) at (10*\xf,-4*\yf) {};
\node[gnode] (w9) at (12*\xf,-5*\yf) {};

\draw[dotted](w5)--++ (-\xf,-\yf);
\draw[dotted](w5)--++ (\xf,-\yf);
\draw[dotted](w6)--++ (-\xf,-\yf);
\draw[dotted](w6)--++ (\xf,-\yf);
\draw[dotted](w9)--++ (-\xf,-\yf);
\draw[dotted](w9)--++ (\xf,-\yf);

\draw[bedge](r)--(w1);
\draw[bedge](w1)--(w2);
\draw[bedge](w2)--(w3);
\draw[bedge](w2)--(w4);
\draw[bedge](w3)--(w5);
\draw[bedge](w3)--(w6);
\draw[bedge](w4)--(w8);
\draw[bedge](w8)--(w9);

\draw [solid] plot [smooth cycle,tension=0.0] coordinates { (-6*\xf,-0.2*\yf) (-12.5*\xf,-7*\yf) (0.5*\xf,-7*\yf)};
\draw [loosely dashed] plot [smooth cycle,tension=0.3] coordinates { (-5.3*\xf,-1.5*\yf) (-6.7*\xf,-1.5*\yf)  (-5.2*\xf,-6.5*\yf)(-0.5*\xf,-6.5*\yf)};
\draw [loosely dashed] plot [smooth cycle,tension=0.3] coordinates { (-5.3*\xf,-1.5*\yf) (-6.7*\xf,-1.5*\yf)  (-11.5*\xf,-6.5*\yf)(-6.8*\xf,-6.5*\yf)};
\node[lnode] (l0) at (-6*\xf,-0.8*\yf){$S$};	
\node[lnode] (l0) at (-9*\xf,-6.4*\yf)  {$S^{(1)}$};	
\node[lnode] (l0) at (-2*\xf,-6.4*\yf)  {$S^{(2)}$};	

\draw [solid] plot [smooth cycle,tension=0.0] coordinates { (1.7*\xf,-3.1*\yf) (-0.5*\xf,-5.5*\yf) (3.9*\xf,-5.5*\yf)};
\draw [solid] plot [smooth cycle,tension=0.0] coordinates { (6.3*\xf,-3.1*\yf) (4.1*\xf,-5.5*\yf) (8.5*\xf,-5.5*\yf)};
\draw [solid] plot [smooth cycle,tension=0.0] coordinates { (12*\xf,-4.1*\yf) (9.8*\xf,-6.5*\yf) (14.2*\xf,-6.5*\yf)};
\end{tikzpicture}
\caption{Example in which algorithm \algDivEx in iteration $t$ divides the considered subtree $S$ into two subtrees $S^{(1)}$ and $S^{(2)}$. The tree $T^R_t$ that connects the roots of the subtrees in $\mathcal{T}_t$ is the subtree containing all thick edges.}
\label{fig-example-algorithm}
\end{center}
\end{figure}

With these definitions, we are now ready to explain the idea of the algorithm~\algDivEx:
During the run of the algorithm, we maintain a set $\mathcal{T}$ of edge-disjoint subtrees of $T$, initially just containing $T$. An example is shown in Fig.~\ref{fig-example-algorithm}, where the triangles show the subtrees that are currently contained in the set $\mathcal{T}$. In every iteration, we consider a subtree $S$ which contains an unexplored vertex and has the highest root, i.e., minimizes $d(r_S)$. 
As long as the leftmost unexplored vertex $v_L$ in $S$ is not too far away from $r_S$, i.e., $d(v_L)-d(r_S)$ is sufficiently small, we send an 
agent to $v_L$ and let it continue the L-DFS from there. 
We do the same if $v_R$ is not too deep and then let the agent continue the R-DFS from $v_R$.
The intuition is that the energy spent to reach $r_S$ is unavoidable, but also the agents in the offline optimum~\opt need to spend this energy without exploring new vertices after the tree has been explored up to depth $d(r_S)$. Thus, the agent only potentially wastes energy to reach $v_L$ (or $v_R$), but from then on explores many new vertices because an agent doing $2m$ edge traversals on a DFS visits at least $m$ distinct vertices. 
If both $v_L$ and $v_R$ are sufficiently deep, we split $S$ into two edge-disjoint subtrees $S^{(1)}$ and $S^{(2)}$, as shown in Fig.~\ref{fig-example-algorithm}. In this case both $S^{(1)}$ and $S^{(2)}$ contain a sufficiently long part of the L-DFS sequence, which has not been covered by any agent. This is important because we want to avoid that an agent is sent to a new subtree which only needs little more exploration. A complete description of~\algDivEx is given in Algorithm~\ref{alg-div-expl}.

\begin{algorithm}[h]
\KwIn{tree $T$ with root~$r$, set of agents $\agentset$, energy bound $B$}
{
$\mathcal{T}=\{T\}$ \;
L-DFS($T,r$) \;
R-DFS($T,r$) \;
\While{$T$ contains unexplored vertex and $\exists$ agent at $r$} {
	\tcp*[l]{move down the roots of the subtrees in $\mathcal{T}$ if possible}
	\ForAll{$S \in \mathcal{T}$ containing an unexplored vertex}{
		$r_0:=r_S$	\;
		\While{$r_0$ only has one child $v$ leading to an unexplored vertex  \\ 
					\hspace*{1cm}\textbf{and} $r_0$ has no unexplored child}{
			$r_0:=v$ \;
		}	
		$\mathcal{T}:=(\mathcal{T}\setminus\{S\})\cup \{T(r_0)\}$ \\
	}
	\tcp*[l]{explore or split the subtree with the highest root}
	$S:=$ subtree in $\mathcal{T}$ that contains an unexplored vertex and minimizes $d(r_{S})$ \;
 	$v_L := $ leftmost unexplored vertex in $S$ \;
 	$v_R := $ rightmost unexplored vertex in $S$  \;
	\uIf{$d(v_L)-d(r_S) \leq \max\{1,1/3 \cdot (B - d(r_S)) \} $}{
		L-DFS($S,v_L$) \;
	} \uElseIf{$d(v_R)-d(r_S) \leq \max\{1,  1/3 \cdot (B - d(r_S))\} $} {
		R-DFS($S,v_R$) \;
	} \Else {
		$v:=$ child of $r_S$ leading to $v_R$ \;
		$S^{(1)}:=$ induced subtree of $S$ containing all vertices not in $T(v)$ \;
		$S^{(2)}:=$ induced subtree of $S$ containing all vertices in $T(v)$ and $r_S$ \;
		$\mathcal{T}:=( \mathcal{T}\setminus \{S\} ) \cup \{S^{(1)},S^{(2)}\}$ \;
		R-DFS($S^{(1)},r_S$) \;
		L-DFS($S^{(2)},r_S$) \;
	}
}
}
\caption{\algDivEx}
\label{alg-div-expl}
\end{algorithm}

In the remainder of this section, we analyze Algorithm~\algDivEx in
order to show that it is $3$-competitive.  Note that the first agent
in \algDivEx simply performs a depth-first search and explores at
least $B/2$ vertices or completely explores the tree. Consequently, if
$k=1$ or if $n<B$, the algorithm is $2$-competitive, and thus we
assume in the following that $n \geq B$ and $k\geq 2$.

For the analysis of \algDivEx, we further need the following
notation. For every iteration~$t$ of the outer while-loop, we let
$k_t\in \{1,2\}$ be the number of agents used by \algDivEx in this iteration and
$k_0=2$ be the number of agents used before the first iteration of the
outer while-loop. Further, let $\mathcal{T}_t$ be the set of subtrees
$\mathcal{T}$ at the end of iteration~$t$ and let $T^R_t$ be the
unique subtree of $T$ that connects the set of roots $\{ r_S \mid S
\in \mathcal{T}_t \}$ of all subtrees with the minimum number of
edges.  Moreover, we denote the subtree $S$ with the highest root
considered by \algDivEx in iteration~$t$ by $S_t$ and its root by
$r_t$. Finally $\bar{t}$ denotes the total number of iterations of the
while-loop.

The crux of our analysis is to show that the amortized amount of
energy spent making progress on the L-DFS or R-DFS is $\tfrac{2}{3}
\cdot k_i \cdot (B-d(r_i))$ for the agents in iteration $i$, as stated
in the following lemma.

\begin{lemma}\label{lem_bound_ldfs_covered}
The algorithm \algDivEx either completely explores $T$ or all agents
used by the algorithm together cover at least
\begin{align*}
\tfrac{2}{3}(|T^R_{\overline{t}}|-1)  +  \sum_{0\leq i \leq \overline{t}} \tfrac{2}{3} \cdot k_i \cdot (B-d(r_i))
\end{align*}
 distinct edges of  the total L-DFS sequence of $T$.
\end{lemma}

\begin{proof}
Let us assume that  \algDivEx does not completely explore $T$ and let
 $\mathcal{U}_t $ be the subset of $\mathcal{T}_t$ containing all subtrees with an unexplored vertex. We will show by induction over $t$ that 
 all agents used by \algDivEx up to the end of iteration~$t$ together cover at least 
\begin{align}
\tfrac{2}{3}(|T^R_t|-1) + \sum_{S\in \mathcal{U}_t} \tfrac{2}{3} (B - d(r_S)) +  \sum_{0\leq i \leq t} \tfrac{2}{3} \cdot k_i \cdot (B-d(r_i))
\label{lb-edges-covered-upto-iteration-t}
\end{align}
 distinct edges of  the total L-DFS sequence of $T$. It may happen that in the last iteration~$\overline{t}$ of \algDivEx the third case occurs, but only one agent is left at the root. We will treat this special case separately at the end of the proof. First, we show the lower bound above for all $t$, for which iteration~$t$ is completed, i.e., there are enough agents for \algDivEx to finish iteration~$t$.
 
For $t=0$, we have $\mathcal{U}_0=\{T\}$ as \algDivEx does not completely explore $T$ by assumption, $k_0=2$, $r_0=r_T$, and $T^R_t$ only contains $r_T$. 
Thus the lower bound~\eqref{lb-edges-covered-upto-iteration-t} on the number of edges covered by the first two agents evaluates to $2 B$. 
The first agent used by \algDivEx performs an L-DFS and covers exactly $B$ edges of the total L-DFS sequence of $T$. The second agent performs an R-DFS starting at the root of $T$ and also covers exactly $B$ edges of the total L-DFS sequence of $T$. The edges covered by the second agent are distinct from the edges covered by the first because $T$ is not completely explored by the algorithm by assumption. Hence, the lower bound~\eqref{lb-edges-covered-upto-iteration-t} holds for $t=0$.

Now, assume that the lower bound~\eqref{lb-edges-covered-upto-iteration-t} holds for $t-1$. We will show it for iteration~$t$. 
Let $\mathcal{U}'_{t-1}$ be the set of subtrees $\mathcal{U}_{t-1}$ after the for-all loop in iteration~$t$ terminated and possibly some roots of the trees in $\mathcal{U}_{t-1}$ were moved down. We claim that
\begin{align}
\tfrac{2}{3}(|T^R_{t-1}|-1) + \sum_{S \in \mathcal{U}_{t-1}}  \tfrac{2}{3} (B - d(r_S)) = \tfrac{2}{3}(|T^R_{t}|-1) + \sum_{S \in \mathcal{U}'_{t-1}} \tfrac{2}{3} (B - d(r_S)).
\label{equality-tr-ut}
\end{align}
For any subtree $S \in \mathcal{U}_{t-1}$, let $S'\in \mathcal{U}'_{t-1}$ be the corresponding subtree after the root of $S$ was possibly moved down. The tree $T^R_{t}$ contains all vertices of the tree $T^R_{t-1}$ plus the path from $r_S$ to $r_{S'}$, i.e.,  $d(r_S)-d(r_{S'})$ additional vertices, for all  $S \in \mathcal{U}_{t-1}$. This already implies~\eqref{equality-tr-ut}.

Applying~\eqref{equality-tr-ut} on the lower bound~\eqref{lb-edges-covered-upto-iteration-t} for $t-1$ yields that the number of edges  of the total L-DFS sequence of $T$ covered by the agents up to \mbox{iteration~$t-1$} is at least
\begin{align}
\tfrac{2}{3}(|T^R_{t}|-1) + \sum_{S \in \mathcal{U}'_{t-1}} \tfrac{2}{3} (B - d(r_S)) +  \sum_{0\leq i \leq t-1} \tfrac{2}{3} \cdot k_i \cdot (B-d(r_i)). \label{num-edges-covered-t-1-version2}
\end{align}

Let now $S_t$ be the subtree with root~$r_t$ considered by the algorithm in iteration~$t$ as defined above and $v_L$, $v_R$ be defined as in the algorithm. 

First, assume that we have $d(v_L)-d(r_t) \leq \max\{1, 1/3 \cdot (B -
d(r_t))\}$ and let $A_0$ be the only agent used by the algorithm in
iteration~$t$. Note that if $1/3 \cdot( B-d(r_t)) < 1$, then once it has
reached $r_t$, agent $A_0$ has either one or two energy left. In the
first case, $A_0$ only explores $v_L$ and makes a progress of $1$ on
the total L-DFS sequence. In the second case, $A_0$ makes a progress
of $2$ on the total L-DFS sequence: it goes to $v_L$ and then either it
visits a child of $v_L$, or it goes back to $r_t$. Consequently, if
$1/3 \cdot (B-d(r_t)) < 1 = d(v_L)-d(r_t)$, $A_0$ makes a progress of at
least $(B-d(r_t)) \geq 2/3\cdot (B - d(r_t))$ on the total L-DFS
sequence.

Suppose now that $1 \leq d(v_L)-d(r_t) \leq 1/3 \cdot (B -
d(r_t))$.  Agent~$A_0$ moves to $v_L$ using at most $1/3 \cdot
(B-d(r_t))$ energy and then performs an L-DFS. If $A_0$ does not
completely explore $S_t$, then the set of edges traversed by $A_0$
starting in $v_L$ and directed in the direction the edge is traversed
by $A_0$ has not been covered by any other agent. Therefore $A_0$
makes a progress of at least $2/3 \cdot (B-d(r_t))$ edges on the total
L-DFS sequence.  Adding this progress of agent~$A_0$ to the lower
bound in~\eqref{num-edges-covered-t-1-version2} on the number of edges
covered by the agents in the first $t-1$~iterations and using
$\mathcal{U}_t=\mathcal{U}'_{t-1}$ yields the lower
bound~\eqref{lb-edges-covered-upto-iteration-t} for iteration~$t$.

Next assume that $A_0$ completely explores the subtree $S_t$. We then have $\mathcal{U}_t=\mathcal{U}'_{t-1} \setminus \{S_t\}$ and the lower bound~\eqref{lb-edges-covered-upto-iteration-t} for iteration~$t$ follows directly from the lower bound~\eqref{num-edges-covered-t-1-version2} even if $A_0$ explores only $v_L$ and  only covers two new directed edges of the total L-DFS sequence.

The proof when $d(v_R)-d(r_t) \leq 1/3 \cdot \max\{1, 1/3 \cdot (B -
d(r_t))\}$ is completely analogous. 

Finally, assume that the last case occurs in iteration~$t$ and $S_t$
is split into two subtrees $S^{(1)}$ and $S^{(2)}$ as defined in the
algorithm. Further, let $A_1$ and $A_2$ be the agents
used in iteration~$t$ for performing an R-DFS in $S^{(1)}$ and an
L-DFS in $S^{(2)}$, respectively.

We first show that $v_L$ and $v_R$ are below different children of $r_t$.
Note that we have $d(v_L)-d(r_t) > \max\{1, 1/3 \cdot (B - d(r_t))\}\geq 1$ and $d(v_R)-d(r_t) > \max\{1, 1/3 \cdot (B - d(r_t))\} \geq 1$ and therefore neither $v_L$ nor $v_R$ are children of $r_t$. Suppose, for the sake of contradiction, there is a child $v$ of $r_t$ such that both $v_L$ and $v_R$ are contained in $T(v)$. By the definition of $v_L$ and $v_R$, the subtrees below all other children of $r_t$ must be completely explored. This means $r_t$ only has one child leading to an unexplored vertex. We cannot have $v_L=v_R=v$ as $v_L$ and $v_R$ are not children of $r_t$. But then the root $r_t$ would be moved down to $v$ and possible further at the beginning of iteration~$t$. This is a contradiction. Therefore, $S^{(1)}$ and $S^{(2)}$ are edge-disjoint, non-empty trees and $v_L$ is contained in $S^{(1)}$ and $v_R$ in $S^{(2)}$.

Agent~$A_1$, which moves according to the call R-DFS($S^{(1)},r_t$), moves to $r_t$ using $d(r_t)$ energy and starts an R-DFS making a progress of at least $d(v_L)-d(r_t) > 1/3 \cdot (B - d(r_t))$ on the overall L-DFS sequence, as the part of the L-DFS sequence from $v_L$ to $r_t$ has not been covered by any other agent and has length at least $d(v_L)-d(r_t)$. If $A_1$ does not completely explore $S^{(1)}$, then it makes even a progress of $B-d(r_t)$ on the overall L-DFS sequence.

The second agent used in iteration~$t$, the agent~$A_2$, first moves to $r_t$ using $d(r_t)$ energy and then performs an L-DFS according to the call  L-DFS($S^{(2)},r_t$). We have  $d(v_R)-d(r_t)  > 1/3 \cdot (B - d(r_t))$ and hence $A_2$ makes a progress of at least $1/3 \cdot (B - d(r_t))$ edges  on the overall L-DFS sequence, as the part of the sequence from $r_t$ to $v_R$ has not been covered by any other agent. If $A_2$ does not completely explore $S^{(2)}$, then it also makes a progress of $B-d(r_t)$ on the overall L-DFS sequence.

Let $s \in \{0,1,2\}$ be the number of subtrees among $\{S^{(1)},S^{(2)}\}$ that $A_1$ and $A_2$ do not explore completely. By the above argument, we showed that overall $A_1$ and $A_2$ together make a progress of at least $2/3 \cdot (B - d(r_t)) + s \cdot  2/3 \cdot (B - d(r_t))$
edges on the overall L-DFS sequence of $T$. Adding this progress to
the  lower bound~\eqref{num-edges-covered-t-1-version2} and using $S_t \in \mathcal{U}'_{t-1} \setminus \mathcal{U}_t$ again yields the lower bound~\eqref{lb-edges-covered-upto-iteration-t} for iteration~$t$. 

In order to show the claim, let us consider the last iteration~$\overline{t}$. If \algDivEx can complete this iteration, then the claim follows directly from the lower bound~\eqref{lb-edges-covered-upto-iteration-t} because  $\tfrac{2}{3} (B - d(r_S))\geq 0$ for all $S\in \mathcal{U}_t$ as no agent can explore a vertex below depth $B$ in $T$. Now assume that iteration~$\overline{t}$ is not completed.
But then we have that the number of edges  of the total L-DFS sequence of $T$ covered by the agents up to iteration~$\overline{t}-1$ is at least
\begin{align*}
\tfrac{2}{3}(|T^R_{\overline{t}}|-1) + \sum_{S \in \mathcal{U}'_{\overline{t}-1}} \tfrac{2}{3} (B - d(r_S)) +  \sum_{0\leq i \leq \overline{t}-1} \tfrac{2}{3} \cdot k_i \cdot (B-d(r_i)) 
\end{align*}
by the lower bound~\eqref{num-edges-covered-t-1-version2}. This lower bound already implies the claim, as we have $k_{\overline{t}}=1$ and $\sum_{S \in \mathcal{U}'_{\overline{t}-1}} \tfrac{2}{3} (B - d(r_S)) \geq \tfrac{2}{3} \cdot k_{\overline{t}} \cdot (B-d(r_{\overline{t}}))$. 
\end{proof}

With the lower bound above, we can now prove the main result of this
section.

\begin{theorem}
The algorithm \algDivEx is 3-competitive.
\end{theorem}

\begin{proof}
Assume that the algorithm \algDivEx terminates after iteration~$\bar{t}$. If it completely explores $T$, then it is clearly optimal. So let us assume that it runs out of agents in iteration~$\bar{t}$.

Let $A_1,A_2,\ldots, A_k$ be the sequence of agents used by \algDivEx
in this order and let agent~$A_i$ be used in iteration~$t_i$. We let
$d_i:=d(r_{t_i})$ be the depth of the root of the subtree visited by $A_i$ in iteration~$t_i$. As the algorithm in every iteration chooses the subtree $S$ with an unexplored vertex which minimizes $d(r_{S})$, we have $d_1 \leq d_2 \leq \ldots \leq d_k$. 

Note that every undirected edge $\{v,w\}$ of the tree appears exactly
twice as a directed edge in the total L-DFS sequence of $T$, as
$(v,w)$ and as $(w,v)$. Thus dividing the bound given by
Lemma~\ref{lem_bound_ldfs_covered} by two yields a lower bound
on the number of distinct undirected edges traversed by the agents. As
$T$ is a tree, this number plus 1 is a lower bound on the number of
vertices visited by the agents. Thus, using the notation $T^R$ instead
of $T^R_{\bar{t}}$, we obtain
\begin{align}
|\alg| \geq \tfrac{1}{3}|T^R|+  \sum_{1\leq i\leq k} \tfrac{1}{3} \cdot (B-d_i).\label{ineq-lb-alg}
\end{align}

Let now $A_1^*,\ldots, A_{k}^*$ be the $k$ agents used by an optimal offline algorithm \opt 
and let $d^*_i$ be the maximum depth of a vertex in $T^R$ that is visited by the agent~$A_i^*$. This is well-defined as every agent at least visits the root~$r$ of $T^R$. We assume without loss of generality that $d^*_1 \leq d^*_2 \leq \ldots \leq d^*_{k}$.
As the agent~$A_i^*$ must use at least $d^*_i$ energy to reach a vertex at depth $d^*_i$ in $T^R$, we have
\begin{align}
|\opt|\leq |T^R|+  \sum_{1\leq i \leq k} (B-d^*_i).\label{ineq-ub-opt}
\end{align}

Consider the maximal index $j \in \{1,\dots,k\}$ such that $d_j >d^*_j$.  
If no such $j$ exists, $d_i \leq d_i^*$ holds for all $1\leq i \leq k$. 
This implies $\sum_{i=1}^{k} (B-d^*_i) \leq \sum_{i=1}^{k} (B-d_i)$ and thus also $|\opt|/ |\alg|\leq 3$ by~\eqref{ineq-lb-alg}~and~\eqref{ineq-ub-opt}.

Otherwise,  we have $d^*_1\leq d^*_2 \leq \ldots \leq d^*_j < d_j$. 
Let $T^j_{\alg}$ be the subtree explored by the first $j$ agents used by \algDivEx. 
We claim that all vertices explored by the agents $A^*_1, \ldots, A^*_j$ are contained in $T^j_{\alg}$. 
Assume, for the sake of contradiction, that there is $1\leq i \leq j$ such that agent $A^*_i$ explores a vertex $u$ which is not contained in $T^j_{\alg}$.
At the moment when the agent~$A_j$ is used by \algDivEx, the root $r_S$ of every subtree $S\in \mathcal{T}_{t_j}$ is contained in $T^R_{t_j}$ and it has depth at least $d_j$. Let $S'\in \mathcal{T}_{t_j}$ be the subtree containing $u$. This means that the agent $A^*_i$ must also visit $r_{S'}$ to reach $u$.
But $T^R_{t_j}$ is a subtree of $T^R$ and thus $A^*_i$ visits a vertex in $T^R$ of depth $d(r_{S'})\geq d_j$. This implies $d_i^*\geq d(r_{S'})\geq d_j$ contradicting the initial assumption that $d_i^*< d_j$.
Consequently, the
agents $A^*_1,\ldots, A^*_j$ in \opt only visit vertices in
$T^j_{\alg}$.  But then the first $j$ agents in \opt visit a strict
subset of the vertices visited by the first $j$ agents in \algDivEx.
In this case, we can just replace the agents $A^*_1, \ldots, A^*_j$
and their paths by the agents $A_1,\ldots, A_j$ and their paths in
\algDivEx and $|\opt|$ does not decrease.  By construction and by
maximality of $j$, we then have $d_i \leq d_i^*$ for all $1 \leq i
\leq k$ which again implies the claim. 
\end{proof}

Note that the analysis of \algDivEx is tight as shown in Appendix~\ref{app:tightanalysis}.

\section{A General Lower Bound on the Competitive Ratio} \label{sec:lb}

We first present an easy example showing a lower bound of 2 on the competitive ratio of any online algorithm.

\begin{proposition} \label{prop:lb}
There exists  no $\rho$-competitive online exploration algorithm with $\rho < 2$.
\end{proposition}

\begin{proof}
Given positive integers~$k$ and~$B$, where $B$ is even, consider the star with center $r$, $k$ rays of length~$B$ and $k\cdot B/2$ rays of 
length~$1$.
For every algorithm, the adversary can ensure that no agent that starts at $r$ ever enters a long 
ray: Whenever an agent is at $r$ and decides to follow an unexplored edge, the adversary directs it to a short ray. 
Therefore, every agent can explore at most~$B/2$ edges and all $k$ agents together at most~$k \cdot B/2$ edges as $B$ is even. On the other hand, the offline optimum sends all agents in the long rays and explores~$k \cdot B$ edges. 
\end{proof}

Note that this lower bound only requires that $B$ is even and otherwise works for any choice of parameter $k$ and $B$. For the lower bound of $(5 + 3 
\sqrt{17})/8 \approx 2.17$ on the competitive ratio, we present a sequence of instances where $k$ and $B$ become arbitrarily large.
We initially construct an instance with general parameters and at the end choose the parameters to maximize the competitive ratio that the online 
algorithm can achieve.
The lower bound instances that we construct are trees that contain very long paths and high degree vertices at certain depth in the tree. The length 
of the paths is determined by the online exploration algorithm.

\begin{restatable}{theorem}{theoremLBcompetitiveRatio}
There exists no $\rho$-competitive online exploration algorithm with $\rho<(5 + 3 \sqrt{17})/8 \approx 2.17$.
\label{theo-lb-competitive-ratio}
\end{restatable}

\begin{figure}
\begin{center}

\usetikzlibrary{arrows,intersections}
\tikzstyle{gnode}=[circle,draw,minimum size=2em,scale=0.7]
\tikzstyle{lnode}=[scale=0.65]
\newcommand*{\xf}{0.6}
\newcommand*{\yf}{0.6}
\begin{tikzpicture}[scale=1]
\node[gnode] (r) at (1*\xf,0) {};
\node[lnode] (l1) at (1*\xf,0) {$r$};

\foreach \j in {-5,-2,1,7} {
	\node[gnode] (a\j) at (\j*\xf,-1*\yf)  {};
	\draw (r)--(a\j);
	\draw[ dotted]   (\j*\xf+0.2*\xf,-2.2*\yf)--(\j*\xf+1*\xf,-2.2*\yf);
}
\node[lnode] (l) at (-5*\xf,-1*\yf)  {$v^{(1)}_{1}$};
\node[lnode] (l) at (-2*\xf,-1*\yf)  {$v^{(1)}_{2}$};
\node[lnode] (l) at (1*\xf,-1*\yf)  {$v^{(1)}_{3}$};
\node[lnode] (l) at (7*\xf,-1*\yf)  {$v^{(1)}_{l}$};

\foreach \i in {0,1,2,6} {
	\draw (a-5)-- ++(-1.2*\xf+\i*0.4*\xf,-0.5*\yf)--++(0,-2.5*\yf-\i*-4*0.04+\i*0.06);
}
\foreach \i in {0,1,2,3,6} {
	\draw (a-2)-- ++(-1.2*\xf+\i*0.4*\xf,-0.5*\yf)--++(0,-2*\yf-\i*-2*0.04+\i*0.06);
}
\foreach \i in {0,1,2,3} {
	\draw (a1)-- ++(-1.2*\xf+\i*0.4*\xf,-0.5*\yf)--++(0,-2*\yf-\i*3*0.04+\i*0.06);
}
\foreach \i in {1,2,3,6} {
	\draw (a7)-- ++(-1.2*\xf+\i*0.4*\xf,-0.5*\yf)--++(0,-1*\yf-\i*6*0.04+\i*0.06);
}

\node[lnode] (l) at (-5*\xf,-4*\yf)  {$v^{(2)}_{1}$};
\node[gnode] (v21) at (-5*\xf,-4*\yf)  {};
\draw(a-5)--(v21);
\foreach \i in {1,2,3,6} {
	\draw (v21)-- ++(-1.2*\xf+\i*0.4*\xf,-0.5*\yf)--++(0,-1*\yf-\i*3*0.04+\i*0.06);
}
\draw[ dotted]   (-5*\xf+0.2*\xf,-5*\yf)--(-5*\xf+1*\xf,-5*\yf);

\node[lnode] (l) at (2.2*\xf,-4*\yf)  {$v^{(2)}_{3}$};
\node[gnode] (v23) at (2.2*\xf,-4*\yf)  {};
\draw (a1)-- ++(-1.2*\xf+6*0.4*\xf,-0.5*\yf)--(v23)--++(0,-1.5*\yf);

\node[lnode] (l) at (5.8*\xf,-4.2*\yf)  {$v^{(2)}_{l}$};
\node[gnode] (v2r) at (5.8*\xf,-4.2*\yf)  {};
\draw (a7)-- ++(-1.2*\xf+0*0.4*\xf,-0.5*\yf)--(v2r);
\foreach \i in {1,2,3,6} {
	\draw (v2r)-- ++(-1.2*\xf+\i*0.4*\xf,-0.5*\yf)--++(0,-2.5*\yf-\i*-4*0.04+\i*0.06);
}
\draw[ dotted]   (5.8*\xf+0.2*\xf,-5*\yf)--(5.8*\xf+1*\xf,-5*\yf);

\draw[loosely dotted]  (2.5*\xf,-1*\yf)--(5.5*\xf,-1*\yf) node[lnode,right] {};
\node[scale=1] (l1) at (11*\xf,-1*\yf)  {depth $d_1$};
\node[scale=1] (l2) at (11*\xf,-4*\yf)  {depth $d_2$};
\end{tikzpicture}
\vspace*{-0.6cm}
\end{center}
\caption{Tree for the lower bound on the competitive ratio.}
\label{fig-lower-bound-tree}
\end{figure}
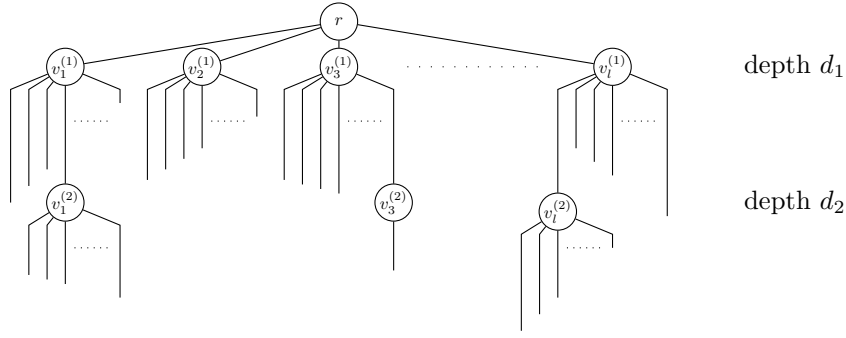

For a given online algorithm \alg, we consider a set of $k= 2 l-1$ agents $\agentset$ for $l \in \mathbb{N}$ with energy $B$ each and we let 
$\Delta={\left\lceil \sqrt{2\cdot  l \cdot B}\right\rceil +2l}$. We now construct a tree~$T$, which is shown in Figure~\ref{fig-lower-bound-tree}, depending on the 
behavior of the algorithm. $T$ has a root~$r$ with $l$ distinct paths, each going from $r$ to a vertex~$\smash{v^{(1)}_i}$ at depth~$d_1$ 
for $i=1,\ldots, l$. Each vertex~$\smash{v^{(1)}_i}$  has degree~$\Delta+1$ and is the root of a subtree~$T_i$. There are $\Delta$~paths connected to every~$\smash{v^{(1)}_i}$ whose length will be determined by the algorithm.
Furthermore, depending on the algorithm, there may exist a vertex~$\smash{v^{(2)}_i}$ at depth~$d_2$ that has  degree~$\Delta+1$ and also $\Delta$~paths connected to it whose length will be determined by the algorithm. We call the subtrees with root~$\smash{v^{(1)}_i}$ and~$\smash{v^{(2)}_i}$ \emph{adaptive subtrees} as they depend on the behavior of the online exploration algorithm.
We further assume that $B$, $d_1$, $d_2$ are even and
\begin{align}
  d_1 + \Delta  < d_2 \leq \tfrac{5}{3} \cdot d_1  \quad \text{ and } \quad 3 \cdot d_1 < B \leq d_1 + 2 \cdot d_2 . \label{assumptions-parameters}
\end{align}

Each of the adaptive trees can be \emph{active}, i.e., as soon as an
agent visits an unexplored vertex on a path another unexplored neighbor
is presented, or \emph{passive}, i.e., all unexplored vertices in the
adaptive tree are leaves. Moreover, every
subtree $T_i$ has a budget $N_i$, which limits the total number of
non-leaf vertices that are presented to the algorithm, i.e., if
$N_i$ vertices that are not leaves have been explored in $T_i$ both
adaptive trees in $T_i$ become passive and from now on all unexplored vertices in $T_i$ are leaves.  The budget $N_i$ is initially $2$ and is increased as
described below when agents enter the subtree $T_i$. Initially every subtree $T_i$ has an active adaptive subtree below $\smash{v^{(1)}_i}$. We now present
new vertices to the algorithm in every subtree $T_i$ for $i \in
{1,\ldots, l}$ according to the following rules:

\begin{enumerate}
\item \emph{When the first agent $A_1$ that has not visited any other tree $T_j\neq T_i$ before enters $T_i$ for the first time:}

The budget $N_i$ of $T_i$ is increased by $(B+d_2)/2-d_1+2\Delta$, the adaptive tree below $\smash{v^{(1)}_i}$ is active and $\smash{v^{(2)}_i}$ has not been discovered. 
The first vertex at depth $d_2$ discovered by $A_1$ is $\smash{v^{(2)}_i}$, i.e., it has degree $\Delta+1$ and is the root of another adaptive tree which is active. 
Additionally, if $A_1$ explores a new vertex $v$ at depth $d>d_2$ in $T_i$ (below~$\smash{v^{(2)}_i}$ or on any branch below~$\smash{v^{(1)}_i}$) 
and the remaining energy of $A_1$ is $\leq d-d_2$, then we stop presenting new vertices on the current path of $A_1$, i.e., $v$ is a vertex without further unexplored neighbors.

\item \emph{When the second agent $A_2$ that has not visited
 any other tree $T_j\neq T_i$ before enters $T_i$ for the first time:} 
\begin{enumerate}
\item \emph{If $A_1$ has explored at most $(d_1+d_2)/2$ vertices in $T_i$:} \\
The adaptive trees both at $\smash{v^{(1)}_i}$ and at $\smash{v^{(2)}_i}$ become passive.
In all following cases below, we assume that $A_1$ explored more than  $(d_1+d_2)/2$ vertices in $T_i$.
\label{case-a1-only-few-vertices}

\item \emph{If $A_1$ has explored the vertex  $\smash{v^{(2)}_i}$ or
  still has enough energy left to reach a vertex $v$ at depth $d_2$
  via an unexplored vertex:} \label{case-a1-reached-v2}
  
If $\smash{v^{(2)}_i}$  has been discovered, the adaptive tree at $\smash{v^{(1)}_i}$ becomes passive, but the adaptive tree at $\smash{v^{(2)}_i}$ remains active. 
If $A_1$ has not visited a vertex at depth $d_2$, then
the adaptive tree at $\smash{v^{(1)}_i}$ becomes passive except for the path via an unexplored vertex to $v^{(2)}_i:=v$ at depth $d_2$, which $A_1$ can reach with its remaining energy. 
From now on, if any agent $A$ is at depth $d>d_2$, then we stop presenting new vertices on the current path of $A$ as soon as the remaining energy is $\leq d-d_2$.

\item \emph{If $A_1$ has not visited a vertex at depth $d_2$ and has not enough energy to reach a vertex at depth $d_2$ via an unexplored vertex:} 

From now on if any agent $A$ is at depth $d >d_1$, we stop presenting new vertices on the current path of $A$ if the remaining energy of $A$ is $\leq d-d_1$.
\end{enumerate}

\item \emph{Whenever an agent $A$ which before has visited a tree $T_j
  \neq T_i$ enters $T_i$ for the first time with remaining energy
  $B_A$:}
  
The budget $N_i$ of $T_i$ is increased by $B_A/2+2$.
If $A$ discovers a vertex $v$ below $\smash{v^{(2)}_i}$  at depth $d>d_2$ and the remaining energy of $A$ is $\leq d-d_2$, then we stop presenting new 
vertices  on this path. Similarly, if
$A$ discovers a vertex $v$ below $\smash{v^{(1)}_i}$ at depth $d> d_1$ (but not on a branch containing $\smash{v^{(2)}_i}$) and the remaining energy of $A$ is $\leq d-d_1$, then we also stop presenting new 
vertices on that path.
\end{enumerate}

Note that in every tree $T_i$, if Case~2b does not occur in $T_i$,
$v_i^{(2)}$ and the adaptive subtree below $v_i^{(2)}$ exist if and only if $A_1$
discovers a vertex $v$ at depth $d_2$.
\medskip

The full proof of the lower bound of $(5 + 3 \sqrt{17})/8 $ on the
competitive ratio is quite technical and given in
Appendix~\ref{app:lb}. Here we want to give some intuition by looking
at two special cases and making some simplifying assumptions, which
do not hold in general. The adaptive trees are constructed in a way
that a path ends exactly when the agent currently exploring that path
has just enough energy to return to $\smash{v^{(1)}_i}$ or
$\smash{v^{(2)}_i}$ respectively. So let us make the simplifying
assumption that the final position of every agent is either at
$\smash{v^{(1)}_i}$ or $\smash{v^{(2)}_i}$ for some $i \in \{1,\ldots,
l\}$. The online algorithm has to balance between sending each agent
to only one subtree~$T_i$ to completely explore it or to move to a
second subtree $T_j$ later to explore more vertices which are close to
the root $r$. We will consider instances with increasing values of $B$
and $l$ in such a way that $l = o(B)$. Note that this implies that
$\Delta = o(B)$.

Let us consider the special case that the algorithm first sends one
agent to each of the subtrees $T_1,\ldots, T_l$ and then a second
agent to every subtree except $T_1$ (there are $2 l-1$ agents and $l$
subtrees). For sake of simplification, assume that $A_1$ visits 
$\smash{v^{(2)}_i}$  and Case 2b occurs in
each subtree $T_i$ when the second agent $A_2$ enters $T_i$. Note that
in this case, $A_1$ cannot visit another subtree as it visits
$\smash{v^{(2)}_i}$ at depth $d_2$ and $2d_2+d_1 \geq B$ by
\eqref{assumptions-parameters}. We further assume
that for each subtree $T_i$, $2 \leq i \leq l$, either the second
agent $A_2$ entering $T_i$ helps $A_1$ to explore $T_i$ completely, or
it goes to $T_1$ to explore new vertices.

The first agent $A_1$ in each subtree $T_i$ can explore at most
$(B+d_2)/2$ vertices in $T$ if its final position is at
$\smash{v^{(2)}_i}$ (it traverses at most $d_2$ edges once and all
other edges are traversed an even number of times) and less vertices
if its final position is at $\smash{v^{(1)}_i}$.  Note that $d_1-2$ of
the vertices explored by $A_1$ are on the path from $r$ to
$\smash{v^{(1)}_i}$ and thus $A_1$ can only explore at most
$(B+d_2)/2-d_1+2$ vertices in $T_i$.  But by construction the budget
$N_i$ is increased by $(B+d_2)/2-d_1+2\Delta$ when $A_1$ enters $T_i$
so that $A_1$ alone cannot deplete the whole budget and completely
explore $T_i$.

As the subtree below $\smash{v^{(1)}_i}$ becomes passive when $A_2$
enters $T_i$, $A_2$ can only explore at most $\Delta$ vertices that
are not below $\smash{v^{(2)}_i}$. Therefore if $A_1$ and $A_2$
completely explore $T_i$, $A_2$ has to go to depth $d_2$ and then it
cannot visit any other subtree as $2d_2+d_1 \geq B$ by
\eqref{assumptions-parameters}. In this case, agents $A_1$ and $A_2$
together then explore at most $N_i$ vertices in $T_i$ plus at most $2
\Delta$ leaves and the path of length $d_1$ leading to $T_i$, i.e.,
they explore at most $(B+d_2)/2+4\Delta+2 = (B+d_2)/2 + o(B)$ vertices.

Suppose now that $A_1$ and $A_2$ do not completely explore the subtree
$T_i$ and that $A_2$ goes to $T_1$ to explore new vertices after
having visited $T_i$. 
Assume that $A_2$ has $B_{A_2}$ energy left when
it enters $T_1$, and note that $B_{A_2} \leq (B-3d_1)/2$ since $A_2$
went first to $T_i$ before entering $T_1$. Agent  $A_2$ can explore at most
$B_{A_2}/2$ new vertices in $T_1$ if its final
position is in $\smash{v^{(1)}_i}$ (every edge it traverses in $T_1$
is traversed an even number of times) and less vertices if its final
position is in $\smash{v^{(2)}_i}$ (since the vertices on the branch
from $\smash{v^{(1)}_i}$ to $\smash{v^{(2)}_i}$ have already been explored). Note that when $A_2$ enters $T_1$, the budget $N_1$ of $T_1$
is increased by $B_{A_2}/2 + 2$ and thus the budget of $T_1$ is never
depleted. As $A_2$ has $B_{A_2}$ energy left when it enters $T_1$ and
spends $3d_1$ energy to first reach $T_i$ and then $T_1$, 
it can have explored at most $(B-3d_1-B_{A_2})/2$ vertices in $T_i$ because
$A_2$ traverses every edge in $T_i$ an even number of times. Overall,
$A_2$ thus explores at most  $(B-3d_1)/2$ new vertices
and $A_1$ at most $(B+d_2)/2$ vertices in this case.

Recall that for sake of simplification, we consider only two strategies for the
online algorithm \alg: either in every tree $T_i$, $2 \leq i \leq l$, $A_1$
and $A_2$ completely explore $T_i$, or for every tree $T_i$, $2 \leq i
\leq l$, the second agent $A_2$ entering $T_i$ also visits $T_1$ (and
$T_i$ is not completely explored by the algorithm). In the first case,
the algorithm explores at most $l \cdot (B+d_2)/2 + o(lB)$
vertices. In the second case, the algorithm explores at most $l \cdot
\left( (B+d_2)/2 + (B-3d_1)/2\right) + o(lB)$ vertices.

Let us now consider an optimal offline algorithm \opt. Whatever the strategy
of \alg is, one can show that there is always an unexplored vertex $u_1$
at depth at most $d_1+\Delta$ in $T_1$. We can assume that $u_1$ has
degree $2l$ and there are $2l-1$ distinct paths of length $B$
connected to it.

If \alg completely explores every tree $T_i$, $2\leq i \leq l$, then
\opt can send all agents to $u_1$ and then each agent explores one of
the paths below $u_1$. In this case, \opt explores at least $B +
(2l-2)\cdot (B-d_1 - \Delta) = 2l \cdot (B-d_1) - o(lB)$ vertices.

If \alg does not completely explore  any $T_i$, $2\leq i \leq l$, then
there exists an unexplored vertex $u_i$ in each tree $T_i$, $2\leq i
\leq l$, and we can assume that there is a path of length $B$
connected to it. In this case, \opt can send an agent to each $u_i$,
$2\leq i \leq l$ that can then explore the path below $u_i$. Then,
\opt can send the remaining $l$ agents to $u_1$ as in the previous
case, and each of these agent explores one of the paths below
$u_1$. In this case, \opt explores at least $lB +
(l-1)\cdot(B-d_1-\Delta) = l\cdot (2B - d_1) - o(lB)$ vertices. 

Since the algorithm can choose the best strategy among the two, we get 
\begin{align*}
  \frac{|\opt|}{|\alg|} \geq \min \left\{\frac{4l \cdot (B- d_1) - o(lB)}{l \cdot (B+d_2) + o(lB)}, \frac{2l \cdot (2B-d_1) - o(lB)}{l \cdot (2 B+d_2-3d_1)+o(lB) }  \right\}.
\end{align*}

In order to maximize the competitive ratio, we want to choose $d_2$ as
small as possible. Because of the initial assumptions on the parameter
in \eqref{assumptions-parameters}, we must have $2d_2+d_1 \geq B$ and
thus we choose $d_2 = (B-d_1)/2$. Additionally, dividing by $l$ and omitting the
terms that vanish as $B$ tends to infinity, we obtain

\begin{align*}
  \frac{|\opt|}{|\alg|} \geq \lim_{B \rightarrow \infty} \min \left\{\frac{8B-8d_1}{3B-d_1}, \frac{8B-4d_1}{5B-7d_1}  \right\}.
\end{align*}

By standard calculus, the competitive ratio is maximized
when the two terms on the right-hand side are equal and this is true
when $d_1 = (19 - 3 \sqrt{17}) B /26 $. These choices of $d_1$ and
$d_2$ satisfy~\eqref{assumptions-parameters} and the above lower bound
evaluates to $(5 + 3 \sqrt{17})/8\approx 2.17$.

We made several simplifying assumptions to get to this bound, but one
can show that no other strategy can beat the lower bound we
established. The challenge in the analysis is that the online
algorithm does not necessarily use one agent after the other, but the
agents may wait in between. This creates many different cases which
need to be grouped and analyzed.

\newpage
\bibliographystyle{plainurl}
\bibliography{treemax-arXiv}

\newpage
\appendix
\section{An Algorithm for Maximal Tree Exploration}
\label{app:algo}

\subsection{L-DFS and R-DFS}\label{app:LR-DFS}

\begin{algorithm}[h]
{
\KwIn{	Tree $T$, starting vertex~$u$ in $T$}
move on a shortest path to $u$ \;
\While{agent~$A$ has energy left and $T$ is not completely explored}{
	\eIf{the subtree below the current node is completely explored}{ 
	         traverse the edge with label 0 \;
	}{
                 traverse the unexplored edge with the smallest label $l>0$ \;
	}
}
}
\caption{L-DFS($T$,$u$)}
\end{algorithm}

\begin{algorithm}[h]
{
\KwIn{	Tree $T$, starting vertex~$u$ in $T$}
move on a shortest path to $u$ \;
\While{agent~$A$ has energy left and $T$ is not completely explored}{
	\eIf{the subtree below the current node is completely explored}{ 
	         traverse the edge with label 0 \;
	}{
                 traverse the unexplored edge with the largest label $l>0$ \;
	}
}
}
\caption{L-DFS($T$,$u$)}
\end{algorithm}

\subsection{\algDivEx is not $\rho$-competitive for $\rho < 3$}\label{app:tightanalysis} 

We construct now an example to show that the analysis of \algDivEx is
tight.  Let $k$, $d \in \mathbb{N}$, $d \geq 2$ and $B= 3 (d-1)$.  The
tree $T$ consists of a root~$r$ connected to $2k$ paths, of which $k$
have length $d$ and $k$ have length $B$, as illustrated in
Fig.~\ref{fig-tree-analysis-tight}. We assume that the edge labels of
the edges incident to the root are increasing from left to right,
i.e., for all $1\leq i \leq 2k-1$, the edge label of $\{r,v_i\}$ is
smaller than the label of $\{r,v_{i+1}\}$.  We further denote the path
$r,v_i,\ldots$ up to the leaf of the tree by $P_i$.

At the beginning of \algDivEx, one agent~$A_1$ performs an L-DFS and completely explores $P_1$ and explores $P_2$ up to depth $d-3$, overall 
exploring $2d-3$ vertices. The second agent~$A_2$ performs an R-DFS and completely explores the rightmost path $P_{2k}$ of length $B$, i.e., 
$B=3(d-1)$ vertices. From now on, in every iteration of the while loop, we have $\mathcal{T}=\{T\}$, $r_S=r$, $d(v_L)=d-2$ and thus 
\begin{align*}
d(v_L)-d(r_S)=d-2 \leq d-1 = 1/3 \cdot (B-d(r_S)).
\end{align*}
This means that, for $i\geq 3$, the agent~$A_i$ used in the iteration~$i-2$ of the outer while-loop, first moves to the unexplored vertex at 
depth $d-2$ on the path $P_{i-1}$, then finishes exploring this path, and runs out of energy at depth $d-3$ in $P_i$. Thus, $A_i$ explores exactly 
$d$ vertices. Overall, the number of vertices explored by the algorithm is therefore
\begin{align*}
2d-3 + 3(d-1) + (k-2) d = 5d - 6 + (k-2) d.
\end{align*}
The optimal offline algorithm sends one agent down each of the paths $P_{k+1}, \ldots, P_{2k}$ exploring $3 k (d-1)$ vertices. 
Hence, we obtain the following lower bound on the competitive ratio:
\begin{align*}
\frac{|\opt|}{|\alg|}= \frac{3 k (d-1)}{5d - 6 + (k-2) d} \xrightarrow{ d \to \infty,k \to \infty } 3. 
\end{align*}

\begin{figure}
\begin{center}
\usetikzlibrary{arrows,intersections}
\tikzstyle{gnode}=[circle,draw,minimum size=2em,scale=0.7]
\tikzstyle{lnode}=[scale=0.65]
\newcommand*{\xf}{0.5}
\newcommand*{\yf}{0.55}
\begin{tikzpicture}[scale=1]
\node[gnode] (r) at (0,0) {};
\node[lnode] (l0) at (0,0)  {$r$};	

\foreach \j in {-8,-6,-3,0,2,4,8} {
	\node[gnode] (a\j) at (\j*\xf,-1.5*\yf)  {};	
	\draw (r)--(\j*\xf,-0.7*\yf)--(a\j);
}
\node[lnode] (l1) at (-8*\xf,-1.5*\yf)  {$v_1$};	
\node[lnode] (l2) at (-6*\xf,-1.5*\yf)  {$v_2$};	
\node[lnode] (l3) at (-3*\xf,-1.5*\yf)  {$v_{i}$};	
\node[lnode] (l3) at (0*\xf,-1.5*\yf)  {$v_k$};	
\node[lnode] (l4) at (2*\xf,-1.5*\yf)  {$v_{k+1}$};	
\node[lnode] (l5) at (4*\xf,-1.5*\yf)  {$v_{k+2}$};	
\node[lnode] (l6) at (8*\xf,-1.5*\yf)  {$v_{2k}$};	

\node[scale=1] (l6) at (-10*\xf,-6*\yf)  {depth $d$};	
\node[scale=1] (l6) at (10*\xf,-8*\yf)  {depth $B$};	
\node[lnode] (l6) at (-3*\xf,-3*\yf)  {$A_i$};	

\foreach \j in {-8,-6,0} {
	\node[gnode] (b\j) at (\j*\xf,-6*\yf)  {};
	\draw[loosely dashed] (a\j)--(b\j);
}

\node[gnode] (v1) at (-3*\xf,-3*\yf)  {};
\node[gnode] (v2) at (-3*\xf,-4*\yf)  {};
\node[gnode] (v3) at (-3*\xf,-5*\yf)  {};
\node[gnode] (b-3) at (-3*\xf,-6*\yf)  {};
\draw[loosely dashed] (a-3)--(v1);
\draw (v1)--(v2);
\draw (v2)--(v3);
\draw (v3)--(b-3);

\foreach \j in {2,4,8} {
	\node[gnode] (b\j) at (\j*\xf,-8*\yf)  {};
	\draw[loosely dashed] (a\j)--(b\j);
}

\draw[loosely dotted] (-5*\xf,-1.5*\yf)--(-4*\xf,-1.5*\yf);
\draw[loosely dotted] (-2*\xf,-1.5*\yf)--(-1*\xf,-1.5*\yf);
\draw[loosely dotted] (5*\xf,-1.5*\yf)--(7*\xf,-1.5*\yf);

\draw[loosely dotted] (-5*\xf,-6*\yf)--(-4*\xf,-6*\yf);
\draw[loosely dotted] (-2*\xf,-6*\yf)--(-1*\xf,-6*\yf);
\draw[loosely dotted] (5*\xf,-8*\yf)--(7*\xf,-8*\yf);
\end{tikzpicture}
\caption{Instance showing that the analysis of \algDivEx is tight.}
\label{fig-tree-analysis-tight}
\end{center}
\end{figure}

\section{A General Lower Bound on the Competitive Ratio}
\label{app:lb}

For every vertex $v$ in $T$, we say that $v$ is \emph{explored} by an agent $A$, if $A$ is the first agent visiting $v$. If $v_i^{(2)}$ is defined, then we say that every vertex on the path from $v_i^{(1)}$ to $v_i^{(2)}$ is explored by the first agent $A_1$, which enters $T_i$ and has not visited any other tree $T_j \neq T_i$ before. It may be even the case that $A_1$ never visits these vertices, but
to simplify the analysis, we will still attribute them to~$A_1$. 

For $i \in \{1,\ldots,l\}$, we let $\agentset_{1,i}$ be the set of agents for which $T_i$ is the first tree they visit and let $\agentset_{2,i}$ be 
the set of agents for which $T_i$ is the second tree they visit, i.e., every agent $A \in \agentset_{2,i}$ has visited a subtree distinct from $T_i$ before.  
Note that an agent can visit at most two subtrees as 
\begin{align}
5 \cdot  d_1 \geq d_1 + 4 \cdot \tfrac{3}{5} d_2 > d_1 + 2 \cdot d_2 \geq B \label{eq-upper-bound-d1-B}
\end{align}
by our assumptions on the parameters in~\eqref{assumptions-parameters}. 
Therefore an agent $A \in  \agentset$ can be contained in one set  $\agentset_{1,i}$ and possible in some other set $\agentset_{2,j}$ for $j  \in \{1,\ldots, l\}\setminus 
\{i\}$.
For every agent $A \in \agentset$  we let $B_A$ denote the remaining energy when $A$ enters a second subtree. If $A$ only enters at most one of the subtrees $T_1,\ldots, T_l$, we set $B_A=0$.
We now establish the following important properties for the number of vertices that the agents explore.

\begin{lemma}\label{lem-properties-lower-bound-instance}
Let $T_i$ be a subtree of $T$ as defined above. 
\begin{enumerate}
\item $B_A \leq B - 3 d_1 $ for all $A \in \agentset$. \label{lem-properties-lower-bound-instance-1}
\item If Case~2b or Case~2c occurs, then the first agent $A_1$ in
  $\agentset_{1,i}$ entering $T_i$ does not visit any other subtree,
  i.e., $B_{A_1} = 0$.\label{lem-properties-lower-bound-instance-2}
\item Every agent $A \in \agentset_{2,i}$ explores at
  most $B_A/2+2$ vertices in $T_i$. \label{lem-properties-lower-bound-instance-3}
\item The first agent $A_1$ in $\agentset_{1,i}$ entering $T_i$
  explores at most $(B +d_2)/2-d_1+2\Delta$ vertices.\label{lem-properties-lower-bound-instance-4}
\item If $|\agentset_{1,i}|\leq 1$, then the agents in  $\agentset_{1,i} \cup \agentset_{2,i}$
visit strictly less than $N_i$ vertices in~$T_i$.  \label{lem-properties-lower-bound-instance-5}
\item If the adaptive tree below $\smash{v_i^{(1)}}$ is active and the budget $N_i$ is not depleted, then there is an unexplored vertex in $T_i$ at depth at most $d_1 + \Delta$.\label{lem-properties-lower-bound-instance-6}
\end{enumerate}
\end{lemma}

\begin{proof}
  
\begin{enumerate}
\item Note that we have $B- 3 d_1 > 0$ by our initial assumptions on the parameters in~\eqref{assumptions-parameters} and thus the claim trivially holds if $A$ visits at most one of
  the subtrees $T_1,\ldots, T_l$, i.e., if $B_A = 0$.  Now, consider
  an agent $A \in \agentset$ visiting two subtrees and assume without
  loss of generality, that $A$ first visits $T_1$ and afterwards
  enters $T_2$ with remaining energy $B_A$. To reach $T_1$ the agent
  needs to traverse $d_1$ edges. Afterwards to reach $T_2$, the agent $A$ needs   to traverse another $2d_1$ edges. Thus, we must have $B_A \leq B - 3   d_1$.

\item In both cases, agent $A_1$ has explored more than $(d_1+d_2)/2$
  vertices in $T_i$. If $A_1$ visits another subtree it traverses
  every edge in $T_i$ an even number of times and therefore needs at least $d_1+d_2$ energy to explore more than $(d_1+d_2)/2$
  vertices. Moreover, $3 d_1$ energy is needed to first reach $T_i$
  and then another subtree. As $3 d_1 + (d_1+ d_2) > 5d_1 \geq B$ by~\eqref{assumptions-parameters} and~\eqref{eq-upper-bound-d1-B}, $A_1$ cannot visit another subtree.

\item By definition, the remaining energy of the agent $A$ when
 entering $T_i$ is $B_A$. If the final position of $A$ is not in $T_i$, then it traverses every edge in $T_i$ an even number of times and in particular $A$ traverses at most $B_A/2$ edges in $T_i$. These can be incident to at most $B_A/2 + 1$ vertices, which yields the claim.
 
Now, consider the case that the final position of $A$ is below $\smash{v_i^{(1)}}$ and not below 
$\smash{v_i^{(2)}}$ and not on the path between $\smash{v_i^{(1)}}$ and $\smash{v_i^{(2)}}$.
This means that at some point $A$ must have visited a vertex $v$ at depth $d$ with remaining energy exactly $d-d_1$. Recall that $B$ and $d_1$ are even, hence $B_A$ is even and this must happen at some point. Then $A$ has exactly enough energy left to move to $\smash{v_i^{(1)}}$ and, in particular, $A$ cannot reach any other path below $\smash{v_i^{(1)}}$. If $v$ is explored by $A$, then $v$ has no new unexplored neighbor and we can simply assume that $A$ returns to  $\smash{v_i^{(1)}}$ as this does not change the number of neighbors it explores. In this case $A$ has traversed every edge in $T_i$ an even number of times and therefore can have explored at most  $B_A/2 + 1$ vertices. If $v$ is not explored by $A$, then $A$ can only explore at most one more vertex after visiting $v$ with energy $d-d_1$, because the current path ends immediately when $A$ explores a new vertex. Compared to the case that $v$ is explored by $A$, agent $A$ only explores at most one additional vertex in this case so that we can bound the total number of vertices explored by $A$ by  $B_A/2 + 2$.

Next consider the case that the final position of $A$ is on the path between  $\smash{v_i^{(1)}}$ and $\smash{v_i^{(2)}}$. In particular, this implies that $\smash{v_i^{(2)}}$ is defined and all vertices on the path between $\smash{v_i^{(1)}}$ and $\smash{v_i^{(2)}}$ are attributed to $A_1$. Note that then all edges that are not on that path, must be traversed an even number of times by $A$ and we therefore again
obtain that $A$ can explore at most $B_A/2 + 1$ vertices, which yields the claim.

Finally, the case the final position of $A$ is below $\smash{v_i^{(2)}}$ is completely analogous to the
case that the final position is  below $\smash{v_i^{(1)}}$  as all vertices on the path from  $\smash{v_i^{(1)}}$ to  $\smash{v_i^{(2)}}$ are attributed to $A_1$.

\item
Let $A_1$ be the first agent entering $T_i$. If $A_1$ visits another subtree $T_j \neq T_i$ afterwards, then $A_1$ traverses every edge in $T_i$ an even number of times and needs $3d_1$ energy to first reach $T_i$ and afterwards $T_j$.
Overall, $A_1$ can therefore explore at most $(B-3d_1)/2$ vertices in $T_i$ and
as $(B +d_2)/2-d_1+2\Delta \geq (B-3d_1)/2$ this yields the claim. 

 From now on, we can therefore assume that $A_1$ only visits the subtree $T_i$. 
  The energy that $A_1$ spends in $T_i$ is at most $B-d_1$, as $B-d_1$ is the maximum energy possible when entering $T_i$. If the final position of 
the agent $A_1$ is at depth $d_2$ or above, then it traverses at most $d_2-d_1$ edges in $T_i$ once using $d_2-d_1$ energy and exploring at most 
$d_2-d_1+1$ vertices. All other edges in $T_i$ traversed by $A_1$ must be traversed at least twice which means there is at most one explored vertex for every two energy used. Overall, the 
number of explored vertices is thus bounded by
\begin{align*}
(d_2-d_1+1) + \frac{B-d_1 - (d_2-d_1)}{2} = \frac{B+d_2}{2} -d_1 + 1,
\end{align*}
if the final position of $A_1$ is at depth $d_2$ or above. If the final position of $A_1$ is below $d_2$, there has to be a vertex $v$ at depth $d$ 
visited by $A_1$ such that the remaining energy of $A_1$ when visiting $v$ is exactly $d-d_2$ (recall that $d_2$ and $B$ are even by assumption). 
If $v$ is explored by $A_1$, then $v$ is the last vertex that $A_1$ explores because $v$ then is a vertex without further neighbors and $A_1$ cannot reach another path below $\smash{v_i^{(1)}}$ or $\smash{v_i^{(2)}}$. If $v$ has been already explored by another agent, then $A_1$ can only explore one more additional vertex as the path also ends immediately if $A_1$ explores a vertex.
 If $A_1$ after visiting $v$ with remaining energy $d-d_2$, would directly move up towards $\smash{v_i^{(1)}}$, its final position would be at depth $d_2$ and by the argument above $A_1$ could explore at most $(B+d_2)/2 -d_1 + 1$ vertices. As $A_1$ can explore only at most one more vertex, as we just showed, the total number of vertices explored by $A_1$ is bounded by $(B+d_2)/2 -d_1 + 2$ in this case. 

However, in  Case~2b, it can happen that $\smash{v_i^{(2)}}$  is defined as it can be reached by $A_1$ with its remaining energy when $A_2$ enters $T_i$, but $A_1$ does not visit $\smash{v_i^{(2)}}$. Recall that we always attribute the vertices on the path between $\smash{v_i^{(1)}}$ and $\smash{v_i^{(2)}}$  to $A_1$, even if $A_1$ never visits them.
 If $A_1$ visits $\smash{v_i^{(2)}}$, then it visits all vertices on the path
  between $\smash{v_i^{(1)}}$ and $\smash{v_i^{(2)}}$  and
by the argument above the number of vertices visited by $A_1$ is bounded by $(B+d_2)/2 -d_1 + 2$. 
As the adaptive tree at $\smash{v_i^{(1)}}$ becomes passive when $A_2$ enters $T_i$, $A_1$ can from then on only explore $\Delta$ vertices which are not on the path between $\smash{v_i^{(1)}}$ and $\smash{v_i^{(2)}}$ or below $\smash{v_i^{(2)}}$. This means compared to the case that $A_1$ visits  $\smash{v_i^{(2)}}$, $A_1$ can only visit additional $\Delta$ vertices and therefore the overall number of vertices explored by $A_1$ is bounded by $(B+d_2)/2 -d_1 + 2\Delta$ in this case as $2+\Delta \leq 2 \Delta$. This yields the claim.
  
  \item 
By Statement~\ref{lem-properties-lower-bound-instance-3} of the lemma, every agent $A\in \agentset_{2,i}$ entering $T_i$ explores at most $B_A/2+2$ vertices and the budget $N_i$ is also increased by this value 
when $A$ enters $T_i$. Thus, if $\agentset_{1,i} = \emptyset$, the number of vertices explored in $T_i$ will always be less than the budget, as $N_i$ 
is initially $2$. 
Now assume, there is one agent $A_1\in \agentset_{1,i}$ entering $T_i$. By Case~1 in the construction of the lower bound, the budget $N_i$ is increased by $(B+d_2)/2-d_1+2\Delta$ and by 
Statement~\ref{lem-properties-lower-bound-instance-4} of the lemma, $A_1$ also explores at most $(B+d_2)/2-d_1+2\Delta$ vertices in $T_i$. Thus the budget $N_i$, which is initially $2$, is also larger than the number of explored vertices in $T_i$ in this case.   
  
\item  
Suppose, for the sake of contradiction, that the budget $N_i$ is not depleted and the adaptive tree below  $\smash{v_i^{(1)}}$ is active,
but there is no unexplored vertex at depth at most $d_1+ \Delta$ in $T_i$.
Recall that there are $\Delta$ path below $\smash{v_i^{(1)}}$ and $\Delta=\ceil{\sqrt{2\cdot  l \cdot B}}+2l$. We have $2l-1$ agents and each agent 
can be responsible for at most one path to be fully explored and end because the agent has remaining energy $\leq d-d_1$ at depth $d$. If all other 
$\ceil{\sqrt{2\cdot l \cdot B}}+1$ paths are fully explored up to depth $\Delta$, then these path contain at least $\Delta \cdot \ceil{\sqrt{2\cdot  
l \cdot B}} \geq 2\cdot  l \cdot B$ vertices. But all agents together only have $(2\cdot  l -1) \cdot B$ energy and hence cannot visit all these vertices. This is a contradiction.
\end{enumerate} 
\end{proof}
 
 We will say that Case~2a occurs in $T_i$ if $|\agentset_{1,i}|\geq 2$ and Case~2a occurs when the second agent $A_2 \in \agentset_{1,i}$ enters 
$T_i$. Analogously for Case~2b and Case~2c. We partition the subtrees into the following three sets:
\begin{align*}
M_0 &:= \{ i \mid B_A>0 \text{ for all } A \in \agentset_{1,i} \text{ or Case~2a occurs in }T_i \}, \\
M_1 &:= \{ i \mid  T_i \text{ is not completely explored, } \exists A \in A_{1,i} \text{ with } B_A=0 \\
& \quad \quad \quad \   \text{ and Case~2a does not occur} \}, \\
M_2 &:= \{ i \mid  T_i \text{ is completely explored and Case~2b or Case~2c occurs in } T_i \}.
\end{align*}

\begin{lemma}\label{lem-ineq-lower-bound}
Let $T_i$ be a subtree of $T$, $
|T_i|$ be the number of vertices explored in $T_i$ by \alg and $M_0$, $M_1$ and $M_2$ as defined above. 
\begin{enumerate}
\item We have $M_0 \cup M_1 \cup M_2 = \{1,\ldots, l\}$ and $M_i \cap M_j = \emptyset$ for all $i,j \in \{0,1,2\}$ with $i \neq j$.
\item For every $i=1,\ldots,l$, we have 
\begin{align}
 |T_i| \leq  \frac{B+d_2}{2} - d_1 +  6 \Delta + \sum_{A \in \agentset_{2,i}} \frac{B_A}{2}.   \label{bound-ti-general}
\end{align} 
\item If $i \in M_0$, then 
\begin{align}
 |T_i| & \leq  \frac{B+d_2}{2} - d_1 + 4 \Delta +  (|\agentset_{1,i}| -2) \cdot \frac{B- 3 d_1}{2} \notag \\
 & \quad +  \sum_{A \in \agentset_{2,i}} \frac{B_A}{2} - \sum_{A \in \agentset_{1,i}} \frac{B_A}{2}  \label{bound-ti-m0}
\end{align}  
\item If $i \in M_1$, then 
\begin{align}
\sum_{A \in \agentset_{1,i}} B_A \leq (|\agentset_{1,i}|-1)\cdot (B-3 d_1).\label{bound-m1}
\end{align} 
\item If $i \in M_2$, then 
\begin{align}
\sum_{A \in \agentset_{1,i}} B_A \leq (|\agentset_{1,i}|-2)\cdot (B-3 d_1).\label{bound-m2}
\end{align}  
\end{enumerate}
\end{lemma}
 
\begin{proof}
  
 \begin{enumerate}
 \item For the first part of the statement, let $i \in \{1,\ldots,
   l\}\setminus (M_0 \cup M_1\}$, and note that there exists $A \in
   \agentset_{1,i}$ with $B_A = 0$, Case~2a does not occur in $T_i$, and
   $T_i$ is completely explored. By
   Statements~\ref{lem-properties-lower-bound-instance-5}
   and~\ref{lem-properties-lower-bound-instance-6} of
   Lemma~\ref{lem-properties-lower-bound-instance}, we have
   $|\agentset_{1,i}| \geq 2$. Consequently, since Case~2a does not occur in
   $T_i$, necessarily Case~2b or Case~2c occurs in $T_i$ and $i \in M_2$.

We obviously have $M_0 \cap M_1 = \emptyset$ and $M_1 \cap M_2 = \emptyset$. By Statement~\ref{lem-properties-lower-bound-instance-2} of Lemma~\ref{lem-properties-lower-bound-instance}, $B_{A_1}=0$ if Case~2b or Case~2c occurs and thus also $M_0 \cap M_2 = \emptyset$.
\item  The budget $N_i$ of the tree $T_i$, which is initially 2, satisfies
\begin{align*}
 N_i  & \leq 2 + \frac{B+d_2}{2} - d_1 +2 \Delta +   \sum_{A \in \agentset_{2,i}} \left(\frac{B_A}{2}+2 \right)  \\
 & \leq \frac{B+d_2}{2} - d_1 + 4 \Delta + \sum_{A \in \agentset_{2,i}} \frac{B_A}{2},
\end{align*}
where we used $2+ 2 |\agentset_{2,i}| \leq 4 l +2 \leq 2 \Delta$. Since
$T_i$ has at most $2\Delta -1$ leaves, and since the number of
vertices explored in $T_i$ that are not leaves is at most $N_i$, we
have $|T_i| \leq N_i + 2\Delta$. This yields the claim.

 \item 
First we show the claim for the case that $B_A>0$ for all $A \in
\agentset_{1,i}$.  This means that every agent $A \in \agentset_{1,i}$
also visits a second subtree. As $3 d_1$ energy is spent to reach $T_i$ and
afterwards the second subtree and $A$ has still $B_A$ energy left when
entering the second subtree, at most $B-3 d_1 - B_A$ energy is spent
in $T_i$. As every edge in $T_i$ is traversed an even number of times, at most $(B-3 d_1 - B_A)/2$ vertices are explored by $A$ in $T_i$ for all $A \in
\agentset_{1,i}$. Moreover, every agent $A \in \agentset_{2,i}$
explores at most $B_A/2+2$ vertices in $T_i$ by
Lemma~\ref{lem-properties-lower-bound-instance}. Additionally using $2
|\agentset_{2,i}| \leq 2 \Delta$, we thus have
\begin{align*}
|T_i| & \leq \sum_{A \in \agentset_{1,i}} \frac{B-3d_1-B_A}{2} + \sum_{A \in \agentset_{2,i}} \left( \frac{B_A}{2}+2 \right) \\
& = | \agentset_{1,i}| \cdot  \frac{B-3d_1}{2}  + \sum_{A \in \agentset_{2,i}} \frac{B_A}{2} -   \sum_{A \in \agentset_{1,i}} \frac{B_A}{2}+ 2 \Delta.
\end{align*}
We obtain the claim using $ (B+ d_2)/2 - d_1 \geq 2 \cdot (B-3d_1)/2$
as $d_2>d_1$ and $5 d_1 >B$ by~\eqref{assumptions-parameters} and~\eqref{eq-upper-bound-d1-B}.

Now assume Case~2a occurs and let $A_1 \in \agentset_{1,i}$ be the first agent entering $T_i$ and $A_2 \in \agentset_{1,i}$ the second agent 
entering $T_i$. As Case~2a occurs, $A_1$ explores at most $(d_1+d_2)/2$ vertices in $T_i$. If $B_{A_1}>0$, i.e., $A_1$ also enters a second tree, we 
can even bound the number of vertices explored by $A_1$ in $T_i$ by $(B-3 d_1 - B_{A_1})/2$. We have $(d_1+d_2)/2 >(B-3 d_1)/2$ as $d_2>d_1$ and $5 
d_1 >B$ by~\eqref{assumptions-parameters} and~\eqref{eq-upper-bound-d1-B}. Therefore, we can both for $B_{A_1}=0$ and for $B_{A_1}>0$ bound the number of vertices explored by $A_1$ until $A_2$ enters $T_i$ by $(d_1+d_2 - 
B_{A_1})/2$. 
As soon as $A_2$ enters $T_i$ all agents together can only explore the unexplored leaves, i.e., at most $2 \Delta$ vertices.
  Moreover, every agent $A \in \agentset_{2,i}$ explores at most $B_A/2+2$ vertices in $T_i$ by Lemma~\ref{lem-properties-lower-bound-instance}. 
Overall, we hence have
\begin{align*}
|T_i|  & \leq \frac{d_1 + d_2-B_{A_1}}{2} + 2 \Delta + \sum_{A \in \agentset_{2,i}} \left(\frac{B_A}{2}+2\right) \\
& \leq  \frac{d_1 + d_2-B_{A_1}}{2} + 4 \Delta + \sum_{A \in \agentset_{2,i}} \frac{B_A}{2},
\end{align*}
where we again used $ 2|\agentset_{2,i}| \leq 2 \Delta$. We also have $0 \leq B - 3 d_1-B_A$ for all $A \in \agentset_{1,i}$ by 
Lemma~\ref{lem-properties-lower-bound-instance} and obtain
\begin{align*}
|T_i|
& \leq  \frac{d_1 + d_2-B_{A_1}}{2} + 4 \Delta + \sum_{A \in \agentset_{1,i}\setminus \{A_1\}} \frac{B-3d_1-B_A}{2} + \sum_{A \in \agentset_{2,i}} 
\frac{B_A}{2}  \\
& =  \frac{d_1 + d_2}{2} + 4 \Delta + (|\agentset_{1,i}|-1) \cdot \frac{B-3d_1}{2}  - \sum_{A \in \agentset_{1,i}} \frac{B_A}{2} + \sum_{A \in 
\agentset_{2,i}} \frac{B_A}{2} \\
& =  \frac{B + d_2}{2} - d_1 + 4 \Delta + (|\agentset_{1,i}|-2) \cdot \frac{B-3d_1}{2}  - \sum_{A \in \agentset_{1,i}} \frac{B_A}{2} + \sum_{A \in 
\agentset_{2,i}} \frac{B_A}{2}.
\end{align*}
 
\item The bound follows directly from the fact that $B_A=0$ for some $A \in \agentset_{1,i}$ and
 $B_A \leq B - 3 d_1$ for all $A \in \agentset_{1,i}$ by Lemma~\ref{lem-properties-lower-bound-instance}.
 
\item In order to show the bound~\eqref{bound-m2}, we proceed along the following key claims:
\begin{enumerate}
\item The bound~\eqref{bound-m2} follows, if the set of agents $\agentset_{1,i} \setminus \{A_1\}$ together visit at least $(B-3 d_1)/2$ distinct vertices in $T_i$ or if there is an agent in $\agentset_{1,i} \setminus \{A_1\}$  that does not visit another subtree.\label{subclaim1-lemma}
\item The bound~\eqref{bound-m2} holds if Case~2b occurs.
\item For Case~2c, the agents in  $\left( \agentset_{i,1} \setminus \{A_1\} \right) \cup \agentset_{i,2}$ need to visit at least $(B-3 d_1)+ \sum_{A \in \agentset_{i,2}} (B_A/2 +2)$ vertices in $T_i$ for $T_i$ to be completely explored. Some of these vertices  may have already been explored by agent~$A_1$. \label{subclaim3-lemma}
\item Let $V_1$ be the set of vertices visited by~$A_1$. 
Further let $e_2$ be the number of vertices explored by the agents in~$\agentset_{i,2}$ that are not contained in~$V_1$ and $n_2$ be the total number of vertices visited by the agents in~$\agentset_{i,2}$ that are contained in~$V_1$. \label{subclaim4-lemma}
Then it holds that  $e_2 + n_2 /2 \leq \sum_{A \in \agentset_{i,2 }} \left(B_A/2 +2 \right)$. 
\item  The claims~\ref{subclaim3-lemma} and \ref{subclaim4-lemma} yield the bound~\eqref{bound-m2} if Case~2c occurs.\label{subclaim5-lemma}
\end{enumerate}

We now show each of the above claims.
\begin{enumerate}
\item 
  By Statement~\ref{lem-properties-lower-bound-instance-2} of Lemma~\ref{lem-properties-lower-bound-instance}, we know that
  $A_1$ cannot visit another subtree, i.e., $B_{A_1} = 0$, as Case~2b or Case~2c occurs when $A_2$ enters $T_i$.  If there
  exists another agent $A' \in \agentset_{1,i}$ such that $B_{A'} =
  0$, then the claim follows directly from the fact that $B_A \leq B
  - 3 d_1$ for all $A \in \agentset_{1,i} \setminus \{A_1,A'\}$ by
  Lemma~\ref{lem-properties-lower-bound-instance}. So assume 
  that for every $A\in \agentset_{1,i}\setminus \{A_1\}$, $B_A >
  0$ holds, i.e., every agent in $\agentset_{1,i}\setminus \{A_1\}$ visits two subtrees and the agents in $\agentset_{1,i}\setminus \{A_1\}$ together visit at least $(B-3 d_1)/2$ distinct vertices in $T_i$.     
As every agent $A$ in  $\agentset_{1,i}\setminus \{A_1\}$ visits a distinct subtree after $T_i$, $A$ traverses every edge in $T_i$ an even number of times. Thus at least $B-3 d_1$ energy is needed to visit $(B-3 d_1)/2$ distinct vertices. But then we already have
\begin{align*}
\sum_{A \in \agentset_{1,i} \setminus \{A_1\}} B_A \leq (|\agentset_{1,i}|-1)\cdot (B-3 d_1),
\end{align*}
as every agents spends an additional $3 d_1 $ energy to first reach $T_i$ and then the second subtree. This implies~\eqref{bound-m2}.

\item

The budget of $T_i$ is increased by $(B+d_2)/2-d_1+2\Delta$ when $A_1$ enters $T_i$,
  but this is also the maximum number of vertices that $A_1$ can
  explore by Lemma~\ref{lem-properties-lower-bound-instance}. Similarly, for every agent $A \in \agentset_{2,i}$ the
  budget is increased by $B_A/2+2$ and the agent can also explore at
  most $B_A/2+2$ vertices by
  Lemma~\ref{lem-properties-lower-bound-instance}. Note that when
  $A_2$ enters $T_i$, the adaptive tree rooted at $\smash{v_i^{(1)}}$ becomes
  passive, and thus agents not entering $\smash{v_i^{(2)}}$ can collectively
  explore at most $\Delta$ vertices after $A_2$ entered $T_i$. We
  claim that if no agent from $\agentset_{1,i}\setminus\{A_1\}$ enters
  $\smash{v_i^{(2)}}$, then $T_i$ cannot be explored.  Indeed, there are
  $\Delta$ paths starting from $\smash{v_i^{(1)}}$ and $\Delta$ paths starting
  from $\smash{v_i^{(2)}}$.
  When the budget $N_i$ is depleted, the agents must have explored $N_i$
  vertices that are not leaves, and consequently, $|T_i| \geq N_i +
  2\Delta$. Since the agents from $\agentset_{2,i}\cup\{A_1\}$ can
  explore at most $N_i - 2$ vertices, the agents from
  $\agentset_{1,i}\setminus\{A_1\}$ have to explore at least $2\Delta
  + 2$
  vertices in $T_i$. Consequently, at least one agent $A'$ from 
  $\agentset_{1,i}\setminus\{A_1\}$ has to visit $\smash{v_i^{(2)}}$ and thus
  $B_{A'} = 0$ as $d_1 + 2 d_2 \geq B$ by~\eqref{assumptions-parameters}. By Claim~\ref{subclaim1-lemma}, this yields~\eqref{bound-m2}.

\item 
As Case~2c occurs when $A_2$ enters $T_i$,  
agent $A_1$ has not enough energy to reach a vertex at depth $d_2$ via an unexplored vertex. We first show that then $A_1$ never visits a vertex at depth $d_2+1$ (it is clear by assumption that $A_1$ never explores a vertex at depth $d_2$ or below, but $A_1$ could still visit a vertex at depth $d_2+1$ on a path that was explored by another agent). If any agent $A$ from  $\agentset_{i,2}$ explores a vertex $v$ at depth $d_2$ in $T_i$, then it must have spend at least $2d_1$ energy to reach the tree it visited before $T_i$ and then come back to the root and another $d_2$ energy to reach $v$. We have $B - 2 d_1 - d_2 \leq d_2-d_1$ as $d_1 + 2d_2 \geq B$ by~\eqref{assumptions-parameters}. Thus $A$ has at most $d_2-d_1$ energy left when it visits $v$ at depth $d_2$ and the path of $A$ ends by Case~3 in the construction of the lower bound. Therefore, $A_1$ cannot reach any vertex at depth $d_2+1$ on a path that was explored by an agent from $\agentset_{i,2}$ as this path ends at depth $d_2$ at the latest.
Agent $A_1$ also cannot visit a vertex at depth $d_2+1$ that was explored by
an agent in $\left( \agentset_{i,1} \setminus \{A_1\} \right)$ as this vertex would be unexplored at the time $A_2$ enters $T_i$ and we assume that at this point $A_1$ cannot reach an unexplored vertex at depth $d_2$.

This means that $A_1$ never visits any vertex at depth $d_2+1$ and can therefore only completely explore one path below $\smash{v_i^{(1)}}$ containing at most $d_2-d_1+1$ vertices. All other vertices visited by $A_1$ that are not on that path have to be visited by other agents since otherwise there is an unexplored vertex at the end of that path.
For $T_i$ to be completely explored, the budget $N_i$ must be completely depleted as otherwise the adaptive tree below $\smash{v_i^{(1)}}$ remains active and there is an unexplored vertex in $T_i$ by Statement~\ref{lem-properties-lower-bound-instance-6} of Lemma~\ref{lem-properties-lower-bound-instance}. Thus all
$N_i$ vertices, except for at most $d_2-d_1+1$, need to be visited by the agents
in $\left( \agentset_{i,1} \setminus \{A_1\} \right) \cup \agentset_{i,2}$ for $T_i$ to be completely explored. We have 
\begin{align}
N_i - \left(d_2-d_1+1 \right)   \geq \frac{B-d_2}{2}  + \sum_{A \in \agentset_{i,2}}\left( \frac{B_A}{2} +2\right). \label{lb-num-vertices-visited-by-A1-A2}
\end{align}
Using, $d_1 + 2 d_2 \geq B$ and $d_2 \leq  5/3 \cdot  d_1$ by \eqref{assumptions-parameters}, we obtain
\begin{align*}
2 B - 6 d_1 \leq (d_1 + 2 d_2) + B - 6 d_1 = 3d_2  - 5 d_1 + (B - d_2) \leq B- d_2.
\end{align*}
This implies $B- 3d_1 \leq (B- d_2)/2$ and together with~\eqref{lb-num-vertices-visited-by-A1-A2} this yields the claim.

\item 
For an agent $A\in \agentset_{i,2}$, let $e_A$ be the number of vertices in $T_i$ that are explored by $A$ and not visited by $A_1$. Moreover, let  $n_{A}$ be the number of moves performed by agent~$A$  in $T_i$ increasing the distance from $A$ to $\smash{v_i^{(1)}}$ while visiting a new distinct vertex in $V_1$. We show that $e_A + n_A/2 \leq B_A/2 +2$. 
The claim then follows by using  $n_2= \sum_{A \in  \agentset_{i,2}} n_{A}$ and $e_2= \sum_{A \in  \agentset_{i,2}} e_{A}$.

Consider the last time an agent $A \in  \agentset_{i,2}$ visits a vertex $v$ at depth $d$ and exactly has enough energy to move to  $\smash{v_i^{(1)}}$ (as $B$ and $d_1$ are even, this will happen at some point). Note that $A$ cannot reach any other path below $\smash{v_i^{(1)}}$ and that it can explore at most one vertex as any unexplored vertex that $A$ visits will have no further neighbor. 

First, assume $v$ is explored by $A$. By Case~3 in the construction of the lower bound, the current path ends and $v$ is a vertex without further neighbors. We can now assume that $A$ returns to $\smash{v_i^{(1)}}$, as this does not change $e_A$ or $n_A$. Then $A$ has traversed every edge in $T_i$ an even number of times and we have $e_A + n_A \leq B_A/2 +1$ and thus in particular,  $e_A + n_A/2 \leq B_A/2 +1$ as $n_A\geq 0$.

Next, assume that $v$ is not explored by $A$ and also not visited by $A_1$. If $A$ would return to $\smash{v_i^{(1)}}$, then we can again argue that $A$ traverses every edge an even number of times and obtain $e_A + n_A \leq B_A/2$ because now we even know that the edge traversal to $v$ was neither an exploration move nor is $v$ contained in $V_1$. On the other hand, if $A$ does not return to $\smash{v_i^{(1)}}$ from $v$ then it cannot visit any new vertex in $V_1$ as $A_1$ never visits $v$ and therefore also no vertex below $v$. Moreover, $A$ can explore at most one additional vertex because then the current path will end immediately. Overall, we therefore again obtain $e_A + n_A \leq B_A/2 +1$, which yields  $e_A + n_A/2 \leq B_A/2 +1$.

Finally, assume that $v$ is not explored by $A$ but visited by $A_1$. Let $e_A'$ be the number of vertices not visited by $A_1$ and explored by $A$ until the visit of $v$ with remaining energy $d-d_1$ and analogously let $n_{A}'$ be the number of moves performed by agent~$A$ up to that time increasing the distance from $A$ to $\smash{v_i^{(1)}}$ while visiting a new distinct vertex in $V_1$. If $A$
would return to  $\smash{v_i^{(1)}}$ with its remaining energy, it would have traversed every edge
an even number of times and we obtain $e'_A + n'_A \leq B_A/2 +1$.
After visiting $v$ agent $A$ can explore only at most one more vertex as then the path ends immediately. Thus, we have $e_A \leq e_A'+1$.
As $v$ is visited by $A_1$, all vertices between $v$ and $\smash{v_i^{(1)}}$ must also be visited by $A_1$. Hence, it holds that $n'_A \geq d- d_1$.
Moreover, after visiting $v$ agent $A$ only has $d-d_1$ energy left for visiting vertices in $V_1$
implying $n_A-n_A' \leq d-d_1$. Overall, this yields
\begin{align*}
e_A + \frac{n_A}{2} \leq e'_A + 1 + \frac{(d-d_1) + n'_A}{2} \leq e'_A + 1 + \frac{2 n'_A}{2} \leq \frac{B_A}{2}+2. 
\end{align*}

\item 
Let $n_1$ be the total number of vertices in $T_i$ visited by the agents
in $\agentset_{i,1} \setminus \{A_1\}$. We assume $n_1 < (B-3 d_1)/2$ as otherwise the claim follows by Claim~\ref{subclaim1-lemma}. 
First of all, we must have
$n_1+ e_2 \geq \sum_{A \in \agentset_{i,2} } \left(B_A/2 +2\right)$ as
$T_i$ contains at least $N_i+\Delta$ vertices if it is completely explored of which $\sum_{A \in \agentset_{i,2} } \left(B_A/2 +2\right)$ 
are not visited by $A_1$ by Statement~\ref{lem-properties-lower-bound-instance-4} of Lemma~\ref{lem-properties-lower-bound-instance}. Using Claim~\ref{subclaim4-lemma}, this implies
\begin{align}
(B-3 d_1)/2 > n_1 \geq \sum_{A \in \agentset_{i,2} } \left(B_A/2 +2\right)- e_2 \geq n_2/2.\label{inequal-n2}
\end{align}
By Claim~\ref{subclaim3-lemma}, we must further  have
\begin{align}
 n_1 +  n_2 + e_2 \geq B- 3 d_1  +\sum_{A \in \agentset_{i,2}} \left(B_A/2 +2\right) \label{inequal-ea-ba}
\end{align}
for the budget $N_i$ to be depleted and $T_i$ completely explored.
As we have  $\sum_{A \in \agentset_{i,2} }  \left(B_A/2 +2\right) \geq n_2/2 + e_2$ by Claim~\ref{subclaim4-lemma}, we obtain $n_1 +  n_2/2 \geq B- 3 d_1$ from \eqref{inequal-ea-ba}. But this implies $n_1 \geq  (B- 3 d_1)/2$ as $n_2/2< (B-3 d_1)/2$ by~\eqref{inequal-n2}, which is a contradiction.

\end{enumerate}
\end{enumerate} 
 \end{proof}
 
\theoremLBcompetitiveRatio*

\begin{proof}
Let $\alg$ be an online exploration algorithm and let $T$ be the tree defined above, which
 depends on $\alg$ and the parameters $l, d_1, d_2$ and $B$. Assume $t$ of the $l$ subtrees 
 $T_1,T_2 \ldots, T_l$ are completely explored and for $j \in \{1,2,3\}$ let $k_j:=| \bigcup_{i \in M_j } \agentset_{1,i}|$. 
  
  We have $|\alg|  \leq   l \cdot  d_1 + \sum_{i=1}^l |T_i| $, 
  as there  are $l$ paths with $d_1$ edges each connecting the root $r$ to every subtree. We now apply 
Inequality~\eqref{bound-ti-general} from Lemma~\ref{lem-ineq-lower-bound} for all subtrees $T_i$ with $i \in M_1 \cup M_2$ and 
Inequality~\eqref{bound-ti-m0} for all subtrees $T_i$ with $i \in M_0$ and additionally use that $\bigcup_{i=1}^l\agentset_{1,i} \supseteq 
\bigcup_{i=1}^l \agentset_{2,i}$. This yields
\begin{align*}
|\alg| & \leq   l \cdot  d_1 + \sum_{i=1}^l |T_i| \\
& \leq l \cdot d_1 +
\sum_{i \in M_1 \cup M_2} 
 \left( \frac{B+d_2}{2} - d_1 + 6 \Delta +   \sum_{A \in \agentset_{2,i}} \frac{B_A}{2} \right) \\
& \quad + \sum_{i \in M_0} 
 \left( \vphantom{\sum_{A \in \agentset_{2,i}}} \frac{B+d_2}{2} - d_1 +  4 \Delta + (|\agentset_{1,i}| -2) \cdot \frac{B- 3 d_1}{2} \right. \left. +  \sum_{A \in \agentset_{2,i}} \frac{B_A}{2} - \sum_{A \in \agentset_{1,i}} \frac{B_A}{2} \right) \\ 
 & \leq l \cdot \left(\frac{B+d_2}{2} + 6 \Delta \right) + \sum_{i=1}^l 
  \sum_{A \in \agentset_{2,i}} \frac{B_A}{2}  - \sum_{i\in M_0}  \sum_{A \in \agentset_{1,i}} \frac{B_A}{2} \\  & \quad 
  + \sum_{i \in M_0} 
      (|\agentset_{1,i}| -2) \cdot \frac{B- 3 d_1}{2} \\ 
 & \leq l \cdot \left( \frac{B+d_2}{2} + 6 \Delta   \right) + (k_0 -2|M_0|) \cdot \frac{B- 3 d_1}{2} + \frac{1}{2} \sum_{i \in M_1 \cup M_2} 
    \sum_{A \in \agentset_{1,i}} B_A. 
\end{align*}  
Now we can apply the Inequalities~\eqref{bound-m1} and \eqref{bound-m2}. We further use $k_0+k_1 +k_2 \leq k =2l -1$, $|M_0| + |M_1|+ |M_2|=l$, $t\leq 
|M_0| +|M_2|$ and obtain
\begin{align*}
|\alg| & \leq l \left( \frac{B+d_2}{2} + 6 \Delta \right) +   (k_0 -2|M_0|) \cdot \frac{B- 3 d_1}{2}  \\
 & \quad +
\frac{1}{2} \sum_{i \in M_1} \sum_{A \in \agentset_{1,i}} 
 (|\agentset_{1,i}|-1)\cdot (B-3 d_1)\\
 & \quad  +
 \frac{1}{2} \sum_{i \in M_2} \sum_{A \in \agentset_{1,i}}  (|\agentset_{1,i}|-2)\cdot (B-3 d_1) \\
& \leq  l \left( \frac{B+d_2}{2} + 6 \Delta \right) + (k_0+k_1+k_2-2 |M_0|- |M_1| - 2 |M_2|) \frac{B- 3 d_1}{2} \\
& \leq  l \left( \frac{B+d_2}{2} + 6 \Delta\right)  + (l-1 - t) \frac{B- 3 d_1}{2}. 
\end{align*}

Next, we will give a lower bound on the number of vertices explored by an optimal offline algorithm \opt.
As there are $2l-1$ agents and $l$ subtrees, there has to be a subtree $T_i$ with $|\agentset_{1,i}|\leq 1$. Without loss of generality let this subtree 
be $T_1$. By Lemma~\ref{lem-properties-lower-bound-instance}
the subtree $T_1$ then has an unexplored vertex $u_1$ at depth at most $d_1 + \Delta$  and, in particular, is not completely explored, implying $t<l$.

For every subtree $T_i$ that is not completely explored, let $u_i$ be
an unexplored vertex in this tree. We can just assume that every $u_i$
has degree $2l$ and $2l-1$ distinct paths of length $B$ connected to
it. The optimal offline algorithm \opt can then send $l-t$ agents each
to one of the unexplored leaves $u_i$ and then down one of the $2l-1$
distinct paths. These agents in total explore $(l-t) \cdot B$
vertices.  All other $l-1+t$ agents are send to the unexplored vertex
$u_1$ in $T_1$ and then each down one path which is not taken by any
other agent. These agents in total explore at least $(l-1+t) \cdot (B
- d_1 - \Delta) $ vertices. Overall, this yields
\begin{align*}
|\opt|  & \geq  (l-t) \cdot  B + (l-1+t) \cdot (B - d_1 - \Delta) \\
& = (2l-1) \cdot B + (l-1+t) \cdot (-d_1-\Delta). 
\end{align*}
For the competitive ratio, we hence obtain
\begin{align*}
\frac{|\opt|}{|\alg|} \geq \min_{t \in \{0,\ldots, l-1\}} \frac{ (4l-2) \cdot B + (2l-2+2t) \cdot (-d_1-\Delta) }{
l \cdot \left( B+d_2 + 12 \Delta\right)  + (l-1 - t) (B- 3 d_1) }.
\end{align*}
In order to maximize the term on the right-hand side, we want to choose $d_2$ as small as possible. Because of the initial assumptions on the parameters in \eqref{assumptions-parameters}, we must satisfy $2d_2+d_1 \geq B$. We can therefore choose $d_2=(B-d_1)/2$ and get
\begin{align*}
\frac{|\opt|}{|\alg|} \geq \min_{t \in \{0,\ldots, l-1\}} \frac{ (8l-4) \cdot B + (4l-4+4t) \cdot (-d_1-\Delta) }{
l \cdot  \left( 3 B- d_1 + 24 \Delta\right)  + (2 l- 2 - 2 t) (B- 3 d_1) }.
\end{align*}
Note that since we assumed $d_2 \leq 5d_1/3$, we need to have that $B
\leq 13d_1/3$, i.e, $d_1 \geq 3B/13$. We also need to satisfy $3 d_1 < B$ by~\eqref{assumptions-parameters} or equivalently $d_1 < B/3$. 

We now  consider an infinite sequence of instances with the following parameters: For every $i \in \N$, let the energy $B$ of the agents be 
$B^{(i)}:=2^{2 i}$, the parameter $l$ be $l^{(i)}:=2^i$ and the depth $d_1$ be $d^{(i)}_1:=b_1 \cdot B^{(i)}$ for some $b_1 \in (3/13, 1/3)$. Note that 
$d^{(i)}_1$ then satisfies $3 d^{(i)}_1 <  B^{(i)} < 13 d^{(i)}_1 / 3$ as required by our initial assumptions on the parameters. Furthermore, we have
\begin{align*}
\frac{\Delta^{(i)}}{B^{(i)}} = \frac{\ceil{\sqrt{ 2 l^{(i)} \cdot   B^{(i)} }} + 2 l^{(i)} }{B^{(i)}}
\xrightarrow{i \to \infty} 0.
\end{align*}  
By dividing all terms in the numerator and denominator by  $l^{(i)} \cdot   B^{(i)}$ and using the property above, we can compute
\begin{align*}
\frac{|\opt|}{|\alg|} & \geq \min_{t \in \{0,\ldots, l^{(i)}-1\}} \frac{ (8l^{(i)}-4) \cdot B^{(i)} + (4l^{(i)}-4+4t) \cdot (-d^{(i)}_1-\Delta^{(i)}) }{
l^{(i)} \cdot  \left( 3 B^{(i)}- d^{(i)}_1 + 24 \Delta^{(i)}\right)  + (2 l^{(i)}- 2 - 2 t) (B^{(i)}- 3 d^{(i)}_1) } \\
& \xrightarrow{ i \to \infty} \inf_{t \in [0,1)} \frac{8 - 4 b_1 - 4 b_1 \cdot t}{3 - b_1 + 2 - 6 b_1 - 2 t + 6 t \cdot b_1}.
\end{align*}
We still have the freedom to choose $b_1 \in (3/13,1/3)$ to maximize the term on the right-hand side, so we even have
\begin{align*}
\frac{|\opt|}{|\alg|} & \geq \sup_{b_1 \in (3/13,1/3)} \inf_{t \in [0,1)} \frac{8 - 4 b_1 - 4 b_1 \cdot t}{5 - 7 b_1 - 2 t + 6 t \cdot b_1}.
\end{align*}
By standard calculus, we obtain that $b_1= \frac{-3 \sqrt{17}+19}{26} \approx 0.26$ maximizes the infimum and satisfies $3/13 \leq b_1 \leq 1/3$. 
Finally, we get
\begin{align*}
\frac{|\opt|}{|\alg|} & \geq \frac{5 + 3 \sqrt{17}}{8} \approx 2.17.
\end{align*} 
\end{proof}

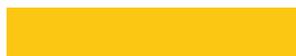
\end{document}